\documentclass[runningheads]{llncs}
\usepackage{ebproof}
\usepackage{amsmath,amssymb,amsfonts}
\usepackage{listings}
\usepackage{physics}
\usepackage[colorlinks,citecolor=violet]{hyperref}
\RequirePackage{xargs}
\RequirePackage{xstring}
\RequirePackage{array}
\RequirePackage{mathtools}
\RequirePackage{thm-restate}
\usepackage{cleveref}
\usepackage{xfrac}
\usepackage{xspace}
\usepackage{xargs}
\usepackage{stmaryrd}
\usepackage{tikz}
\usetikzlibrary{matrix}
\usetikzlibrary{decorations.pathreplacing,calc}
\usepackage[symbol]{footmisc}
\usepackage{multirow}
\usepackage{paralist}
\usepackage{thm-restate}
\usepackage{booktabs}
\usepackage{array}
\usepackage{url}

\newcommand*{\greenleaf}[1]{%
  \tikz[baseline=(X.base)] \node[rectangle, fill=green!25, rounded corners, inner sep=0.4mm] (X) {#1};%
}

\makeatletter
\AtBeginDocument{%
  \def\doi#1{\url{https://doi.org/#1}}}
\makeatother

\newcommand*{\tpkt}{\rlap{$\;$.}}
\newcommand*{\tkom}{\rlap{$\;$,}}

\lstdefinestyle{common}{
  basicstyle=\ttfamily,
  commentstyle=\cmtstyle,
  keywordstyle=\keystyle,
  emphstyle=\emphstyle,
  showstringspaces=false,
  flexiblecolumns=false,
  basewidth={0.5em,0.45em},
  tabsize=2,
  captionpos=b,
  mathescape=true,
  escapechar=\%,
  xleftmargin=0pt,
  xrightmargin=0pt,
  framexleftmargin=0pt,
  abovecaptionskip=\smallskipamount,
  numbers=none,
}

\lstdefinestyle{framed}{
  frame=single,
  xleftmargin=3pt,
  xrightmargin=3pt,
  framexleftmargin=0pt,
  showlines=false,
  rulecolor=\color{gray!25}
}

\lstdefinestyle{qwhile}{
  style=common,
  style=framed,
  language=C
}

\usepackage{pwhile}
\usepackage{wrapfig}

\newcommand{\rgl}{::=}
\newcommand{\defiff}{\mathrel{{:}\!{\iff}}}

\newcommand{\oper}[1]{\mathtt{#1}}
\newcommand{\ope}{\oper{U}}
\newcommand{\meas}[2]{M_{#1,#2}}
\newcommand{\instr}[1]{\mathtt{#1}}
\newcommand{\Variable}[1]{\mathtt{#1}}
\newcommand{\x}{\Variable{x}}
\newcommand{\ct}{\Variable{ct}}
\newcommand{\y}{\Variable{y}}

\newcommand{\n}{\Variable{n}}
\newcommand{\iv}{\Variable{i}}
\newcommand{\kv}{\Variable{k}}
\newcommand{\q}{\Variable{q}}
\newcommand{\qs}{\overline{\q}}
\newcommand{\p}{\Variable{p}}

\newcommand{\asgn}{{:=}}
\newcommand{\expr}[1]{\mathtt{ #1}}
\newcommand{\e}{\expr{e}}
\newcommand{\true}{\instr{tt}}
\newcommand{\false}{\instr{ff}}

\newcommand{\Exp}{\mathtt{Exp}}
\newcommand{\BExp}{\mathcal{B}\mathtt{Exp}}
\newcommand{\NExp}{\mathcal{N}\mathtt{Exp}}

\newcommand{\Cmds}{\mathtt{Stmt}}
\newcommand{\Cmdsk}{\mathtt{Stmt}_\instr{K}}
\newcommand{\Store}{\mathtt{Store}}

\newcommand{\cmd}{\instr{stm}}
\newcommand{\bexp}{\expr{b}}
\newcommand{\nexp}{\expr{n}}

\newcommand{\Qubits}{\mathcal{Q}}
\newcommand{\Bool}{\mathcal{B}}
\newcommand{\Var}{\mathcal{N}}
\newcommand{\K}{\mathcal{K}}
\newcommand{\ltuple}{(}
\newcommand{\rtuple}{)}

\newcommand{\geo}{\fun{GEO}}
\newcommand{\coin}{\mathtt{Cointoss}}
\newcommand{\rus}{\mathtt{RUS}}

\newcommand{\Q}{\mathbb{Q}}

\newcommand{\Comp}{\mathbb{C}}
\newcommand{\RAlg}{\mathbb{A}}
\newcommand{\densop}{\mathfrak{D}(\mathcal{H}_{Q})}
\newcommand{\matrixspace}{\mathcal{M}(\mathcal{H}_{Q})}
\newcommand{\densopalg}{\mathfrak{D}(\tilde{\mathcal{H}}_{Q})}
\newcommand{\matrixspacealg}{\mathcal{M}(\tilde{\mathcal{H}}_{Q})}
\newcommand{\img}{\mathfrak{i}}

\newcommand*{\when}[2]{\ifthenelse{#1}{#2}{}}

\newcommand*{\unlessempty}[2]{\ifthenelse{\equal{#1}{}}{}{#2}}

\newcommand{\up}[1]{\mathrel{{+}_{#1}}}

\newcommand{\CSd}{\mathtt{S}}
\newcommand{\transformer}[4][]{\mathtt{#2}\!\left[\,#3\,\right]^{#1}\!\left\{\textstyle#4\right\}}

\newcommand{\wpt}[3][]{\transformer[#1]{qet}{#2}{#3}}
\newcommand{\qinfer}[3][]{\transformer[#1]{qinf}{#2}{#3}}
\newcommand{\qev}[3][]{\transformer[#1]{qev_\CSd}{#2}{#3}}

\newcommand{\xrightarrowdbl}[2][]{%
  \xrightarrow[#1]{#2}\mathrel{\mkern-14mu}\rightarrow
}
\newcommand{\toop}[1]{\stackrel{}{\to}}
\newcommand{\tomulti}[1]{\xrightarrowdbl{#1}}
\newcommand{\tooqw}{\to_{\!\tiny \textsc{q}}}
\newcommand{\toomqw}[1]{\xrightarrowdbl[]{#1}_{\tiny\textsc{q}}}

\newcommand{\app}{\ }
\NewDocumentCommand\nf{d<>o}{\mathsf{term}^{\IfNoValueF{#1}{\leq#1}}_{#2}}
\NewDocumentCommand\edl{d<>o}{\mathsf{edl}^{\IfNoValueF{#1}{=#1}}_{#2}}
\RenewDocumentCommand\wp{d<>o}{\mathsf{wp}^{\IfNoValueF{#1}{\leq#1}}_{#2}}
\NewDocumentCommand\qwp{d<>o}{\mathsf{qwp}^{\IfNoValueF{#1}{\leq#1}}_{#2}}
\NewDocumentCommand\wpeq{d<>o}{\mathsf{wp}^{\IfNoValueF{#1}{=#1}}_{#2}}
\NewDocumentCommand\qwpeq{d<>o}{\mathsf{qwp}^{\IfNoValueF{#1}{=#1}}_{#2}}
\newcommand{\ect}[1][\cmd]{\mathsf{ect}_{#1}}
\newcommand{\sem}[2][]{\llbracket #2 \rrbracket^{#1}}

\newcommand{\conf}{{\tt Conf}}
\newcommand{\confct}{\conf_{\mathtt{CT}}}

\newcommand{\evalue}[1][]{\mathsf{evalue}^{#1}}
\newcommand{\lmulti}{\{}
\newcommand{\rmulti}{\}}
\newcommand{\size}[1]{| #1 |}
\newcommand{\State}{\instr{St}}

\newcommand{\Dists}{\mathcal{D}}
\newcommand{\supp}{\mathrm{supp}}

\newcommand{\lfp}{\mathsf{lfp}}
\newcommand{\N}{\mathbb{N}}
\newcommand{\Rpos}{\mathbb{R}^{+}}
\newcommand{\Rext}{\mathbb{R}^{+\infty}}
\newcommand\E[2]{\mathbb{E}_{#1}(#2)}

\newcommand{\smx}[1]{\left(\begin{smallmatrix}#1\end{smallmatrix}\right)\!}

\newcommand{\ZBQP}{\textsc{zbqp}}

\newcommand{\qmd}{\instr{P}}
\newcommand{\reduces}{\leq_\mathsf{m}}
\newcommand{\compl}[1]{\ensuremath{\text{co{-}}#1}}
\newcommand{\COF}{\ensuremath{\mathcal{COF}}}
\newcommand{\AST}{\textsc{Ast}}
\newcommand{\PAST}{\textsc{Past}}
\newcommand{\UAST}{\textsc{UAst}}
\newcommand{\UPAST}{\textsc{UPast}}
\newcommand{\WP}[1]{\textsc{Test}_{#1}}
\newcommand{\UWP}[1]{\textsc{UTest}_{#1}}

\makeatletter
\newcommand\newtag[2]{#1\def\@currentlabel{#1}\label{#2}}
\makeatother

\usetikzlibrary{decorations.pathmorphing}
\usetikzlibrary{decorations.markings}
\usetikzlibrary{decorations.pathreplacing}
\usetikzlibrary{arrows}
\usetikzlibrary{shapes}

\pgfdeclarelayer{edgelayer}
\pgfdeclarelayer{nodelayer}
\pgfsetlayers{edgelayer,nodelayer,main}

\tikzset{
cross/.style={path picture={
\draw[-,black](path picture bounding box.north) -- (path picture bounding box.south) (path picture bounding box.west) -- (path picture bounding box.east);
}}}

\tikzstyle{braceedge}=[decorate,decoration={brace,amplitude=10pt}]
\tikzstyle{square box}=[rectangle,fill=white,draw=black,minimum height=6mm,minimum width=6mm,yshift=0.7mm]
\tikzstyle{box}=[rectangle,fill=white,draw=black]
\tikzstyle{dot}=[circle,fill=black,draw=black,inner sep=1.5pt]
\tikzstyle{target}=[{draw,circle,cross,minimum width=0.3 cm}]

\tikzstyle{none}=[inner sep=0pt]
\tikzstyle{empty}=[rectangle,fill=none,draw=none]
\tikzstyle{scaled}=[rectangle,fill=none,draw=none, font=\small]

\tikzstyle{to}=[->,draw=black]
\tikzstyle{naturalto}=[-{Implies},double distance=1.5pt]
\tikzstyle{hook}=[right hook->, draw=black]
\tikzstyle{blueArr}=[->, draw=blue]

\tikzstyle{equal-arrow}=[double equal sign distance]

\tikzstyle{every picture}=[baseline=-0.25em]

\newcommandx{\mparbox}[3][1=l]{%
  \text{\makebox[#2][#1]{\ensuremath{#3}}}}

\newcommand{\tikzrom}[1]{\tikz[overlay,remember picture] \node (#1) {};}
\newcommand*{\AddNotecopy}[4]{%
    \tikz[overlay, remember picture]%
        \draw [decoration={brace,amplitude=0.5em},decorate,ultra thick,black!60]
            ($(#3)!([yshift=0.5ex]#1)!($(#3)-(0,1)$)$) --
            ($(#3)!([yshift=-0.5ex]#2)!($(#3)-(0,1)$)$)
                node [align=center, text width=1.5cm, pos=0.5, anchor=west] {#4};}

\newcommand{\repeatable}[1]{
  \newenvironment{#1*}[2][]
  {\expandafter\restatable[##1]{#1}{##2}\label{##2}}
  {\endrestatable}
}

\newcommand{\again}[1]{\csname#1\endcsname*}

\repeatable{theorem}
\repeatable{lemma}

\makeatletter
\newcommand*{\envskipline}{\hfill\@beginparpenalty=10000}
\newcounter{envnumone}
{\hfill\bgroup\@beginparpenalty=10000%
  \begin{list}{\arabic{envnumone}.}%
    {%
      \usecounter{envnumone}
      \setlength{\itemsep}{0pt}
      \setlength{\topsep}{0pt}
      \setlength{\parsep}{0pt}
      \setlength{\partopsep}{0pt}
      \setlength{\leftmargin}{15pt}
      \setlength{\rightmargin}{0pt}
      \setlength{\itemindent}{0pt}
      \setlength{\labelsep}{5pt}
      \setlength{\labelwidth}{15pt}
    }%
  }%
  {\end{list}\egroup}
  {\end{itemize}\egroup}
\makeatother

\makeatletter
\newcommand{\customlabel}[2]{%
   \protected@write \@auxout {}{\string \newlabel {#1}{{#2}{\thepage}{#2}{#1}{}} }%
}
\makeatother
\newcommandx\law[2][1=]{
  \hypertarget{#1}{
    \begin{array}[c]{>{$\it}l<{$}}
      #2
    \end{array}
  }
  \expandarg
  \def\pattern{\\}
  \def\replace{\ }
  \StrSubstitute{#2}{\pattern}{\replace}[\txt]
  \customlabel{#1}{(\textit{\txt})}
}

\newenvironment{proofcases}
{
\begin{list}{\labelitemi}
{\setlength{\itemsep}{0pt}
 \setlength{\topsep}{0pt}
 \setlength{\parsep}{0pt}
 \setlength{\partopsep}{0pt}
 \setlength{\leftmargin}{10pt}
 \setlength{\rightmargin}{0pt}
 \setlength{\itemindent}{0pt}
 \setlength{\labelsep}{5pt}
 \setlength{\labelwidth}{10pt}
}}
{
 \end{list}
}
\newcommand{\proofcase}[1]{\item[-]\textsc{Case} #1.}

\begin{document}


\title{On the Hardness of Analyzing Quantum Programs Quantitatively}

%

\author{}
\institute{}

\maketitle

\begin{abstract}
In this paper, we study quantitative properties of quantum programs. Properties of interest include (positive) almost-sure termination, expected runtime or expected cost, that is, for example, the expected number of applications of a given quantum gate, etc.
After studying the completeness of these problems in the arithmetical hierarchy over the Clifford+T fragment of quantum mechanics, we express these problems using a variation of a quantum pre-expectation transformer, a weakest pre-condition based technique that allows to symbolically compute these quantitative properties. Under a smooth restriction---a restriction to polynomials of bounded degree over a real closed field---we show that the quantitative problem, which consists in finding an upper-bound to the pre-expectation, can be decided in time double-exponential in the size of a program, thus providing, despite its great complexity, one of the first decidable results on the analysis and verification of quantum programs. Finally, we sketch how the latter can be transformed into an efficient synthesis method.
\end{abstract}


\section{Introduction}
\subsubsection{Motivations.}
Quantum computation is a promising and emerging computational paradigm which can efficiently solve problems considered to be intractable on classical computers \cite{shor,hhl}. 
However, the unintuitive nature of quantum mechanics poses challenging questions for the design and
analysis of corresponding quantum programming. 
Indeed, the quantum program dynamics are considerably more complicated compared to the behavior of classical or probabilistic programs.
Therefore, formal reasoning requires the development of novel methods and tools, a development that has already started and recently gathered momentum in various areas, like \emph{design automation}  \cite{Wille:2019,HYKW:2022}, \emph{programming languages} \cite{qpl,AltenkirchG05,DCM22,JiaKLMZ22,FKRS23}, \emph{verification} \cite{Perdrix08,Chen:2023}, etc.

Among these formal methods, those that allow us to obtain quantitative properties on quantum programs are particularly interesting. They can be used to obtain relevant information about the computations of a quantum program, such as the number of qubits used, the number of unitary operators used. 
Thus enabling the corresponding compiled quantum circuit to be optimized (for example, by minimizing the use of gates that are hard to make fault-tolerant, or by reducing the number of qubits) 
 or to avoid undesirable behavior such as non-termination.
Quantitative properties of interest may also be the question whether or not a program \emph{terminates almost-surely}, that is, whether its probability of non-termination is zero or not.
Similarly, we could aim to capture the \emph{expected values} of (classical) program variables upon termination assumption. The latter can also be employed to reason about the \emph{expected runtime} or the \emph{expected cost} of quantum programs, if we suitable instrument the code with counter variables. 

The program of Figure~\ref{fig:rus} implements a Repeat-Until-Success algorithm that can be used to simulate quantum
unitary operators on input qubit $\q_1$ by using repeated measurements. The quantum step-circuit on the right part corresponds to one iteration of the loop. Variable $\iv$ in the program just acts as a counter for $\oper{T}$-gates. Hence an analysis on the expected value of variable $\iv$ can be used to infer upper bounds on the expected \emph{T-count}, i.e., the expected number of times the so-called $\oper{T}$-gate is used in the full compiled quantum circuit. Such an approach offers the advantage to allow the programmer to implement quantum programs using fewer $\oper{T}$-gates, which are costly to implement fault-tolerantly~\cite{BK05,GKMR14}, and it therefore provides a simple quantum program to illustrate that the study of quantitative properties is paramount.

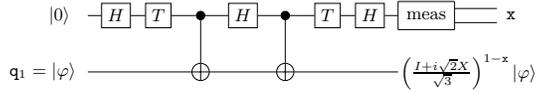
\begin{figure}
\begin{minipage}{0.39\textwidth}
\centering
   $
\rus\ \triangleq  \  \pw {
   \iv^{\Var} <- 0;\\
    \x^\Bool <- \true ; \\
    \WHILE \x  \DO {
           \q_2 <- \ket{0}; \tikzrom{mark1}\quad \quad \quad \tikzrom{mar1}\\
           \q_2  <*  \oper{H} ;\\
           \q_2 <* \oper{T} ;\\
           \iv <- \iv + 1 \\
           \q_2,\q_1 <*  \oper{CNOT} ; \tikzrom{mark2} \quad \quad \quad  \tikzrom{mar2}\\
           \q_2 <*  \oper{H} ;\\
           \q_2,\q_1 <*  \oper{CNOT} ;\\
           \q_2 <*  \oper{T} ;\\
           \iv <- \iv + 1 \\
           \q_2 <*  \oper{H} ; \tikzrom{mark3}\\
           \x  <- \MEAS{\q_2}\quad \quad \quad $ $\tikzrom{mar3}
           }
 }
   $
\end{minipage}
\quad
\begin{minipage}{0.63\textwidth}
\scalebox{0.75}{
\begin{tikzpicture}
	\begin{pgfonlayer}{nodelayer}
		\node [style=none] (0) at (-4, 0.5) {$\ket 0$};
		\node [style=none] (1) at (-4.3, -0.5) {$\q_1 =\ket \varphi$};
		\node [style=none] (2) at (-3.5, 0.5) {};
		\node [style=none] (3) at (-3.5, -0.5) {};
		\node [style=box] (4) at (-3, 0.5) {$H$};
		\node [style=box] (5) at (-2.25, 0.5) {$T$};
		\node [style=target] (6) at (-1.5, -0.5) {};
		\node [style=dot] (7) at (-1.5, 0.5) {};
		\node [style=box] (8) at (-0.75, 0.5) {$H$};
		\node [style=target] (9) at (0, -0.5) {};
		\node [style=dot] (10) at (0, 0.5) {};
		\node [style=box] (11) at (0.75, 0.5) {$T$};
		\node [style=box] (12) at (1.5, 0.5) {$H$};
		\node [style=none] (14) at (2, -0.5) {};
		\node [style=none] (15) at (3.25, -0.5) {$\left(\frac{I + i \sqrt 2 X}{\sqrt 3}\right)^{1-\x} \ket \varphi$};
		\node [style=none] (18) at (4, 0.5) {$\x$};
		\node [style=none] (19) at (2.5, 0.5) {meas};
		\node [style=none] (20) at (2, 0.25) {};
		\node [style=none] (21) at (2, 0.75) {};
		\node [style=none] (22) at (3, 0.75) {};
		\node [style=none] (23) at (3, 0.25) {};
		\node [style=none] (24) at (2, 0.5) {};
		\node [style=none] (25) at (3, 0.625) {};
		\node [style=none] (26) at (3, 0.375) {};
		\node [style=none] (27) at (3.75, 0.625) {};
		\node [style=none] (28) at (3.75, 0.375) {};
	\end{pgfonlayer}
	\begin{pgfonlayer}{edgelayer}
		\draw (2.center) to (4);
		\draw (4) to (5);
		\draw (7) to (6);
		\draw (5) to (7);
		\draw (3.center) to (6);
		\draw (7) to (8);
		\draw (10) to (9);
		\draw (8) to (10);
		\draw (10) to (11);
		\draw (12) to (11);
		\draw (6) to (9);
		\draw (9) to (14.center);
		\draw (20.center) to (23.center);
		\draw (23.center) to (22.center);
		\draw (22.center) to (21.center);
		\draw (21.center) to (20.center);
		\draw (24.center) to (12);
		\draw (27.center) to (25.center);
		\draw (28.center) to (26.center);
	\end{pgfonlayer}
\end{tikzpicture}
}

\end{minipage}
\vspace{-2mm}
\caption{Repeat-until-success program $\rus$ and step-circuit.}
\label{fig:rus}
\end{figure}

In \cite{AMPPZ:LICS:22,LZY22}, new methodologies named \emph{quantum expectation transformers} based on \emph{predicate transformers}~\cite{D76,K85} and \emph{expectation transformers}~\cite{MM05,GKM14} have been put forward to naturally express and study the quantitative properties of quantum programs. However, no attempt was made to automate the corresponding techniques or delineate how complicated such an \emph{automation} could be. Automation of these formal verification techniques in the context of quantum programs is a particularly difficult problem. Indeed, the consideration of Hilbert spaces as  a mathematical framework for describing principles and laws of  quantum mechanics makes it seemingly impossible to reason fully automatically about quantitative properties of quantum program: they involve computational objects of exponential dimensions (in the number of qubits) with scalars ranging over an uncountable domain (i.e., complex numbers). This problem is directly linked to the fact that the set $\Comp$ includes non-computable numbers~\cite{weihrauch2012computable} and that testing the inequality $\leq$ or the equality $=$ of two real numbers is not decidable, even if one restricts their study to computable real numbers.
Consequently, 
the particular nature of quantum programs and of their semantic domain, Hilbert spaces, makes it impossible to directly apply the results obtained in the classical and probabilistic setting~\cite{SchnablS11,KaminskiKatoen}.


%

\subsubsection{Contributions.}
In this paper, we study the hardness of the quantitative properties of mixed classical-quantum programs and provide a first
step towards their (full) automation using quantum expectation transformers.

Towards that end, we restrict the considered quantum gates to the Clifford+T fragment, which is known to be the simplest approximately universal fragment of quantum mechanics~\cite{aaronson2004improved}. 
Clifford+T makes it possible to only consider quantum states with algebraic amplitudes, thus restricting the study to a countable domain.
It implies that our results can accommodate quantum gates employed in actual hardware, recently employed to claim \emph{quantum advantage}, cf~\cite{Arute:2019}.
Moreover, the obtained results are very general as it can be extended to any set of gates with algebraic coefficients.

As motivated, our first contribution is about the general hardness of deciding quantitative properties for mixed classical-quantum programs. For a given input state, we study properties such as \emph{(positive) almost-sure termination}, $(\textsc{P})\AST$ for short, which consists in verifying that a program terminates with probability $1$; \emph{testing problems}, $\WP{\mathcal{R}}$, which consist in comparing a quantum expectation (for example, the mean value of a variable) with a given value (an algebraic and positive real number) wrt the relation $\mathcal{R}$; and the \emph{finiteness problem}, $\WP{\neq \infty}$, which consists in checking that a quantum expectation is finite. For each of those problems, we also study the related \emph{universal problem}, which consists to check the corresponding property for every input.
We establish a precise mapping (Theorem~\ref{thm:mapping}) of the inherent complexity of each problem in the arithmetical hierarchy~\cite{odifreddi1992classical} that is summarized in Table~\ref{fig:mapping}, see Section~\ref{s:mapping}.
E.g., $\AST$ is $\Pi^0_2$-complete while $\PAST$ is $\Sigma^0_2$-complete.
In Table~\ref{fig:mapping}, we have emphasized problems whose hardness has been studied neither in the classical nor in the probabilistic framework by marking them with $(\ddag)$.


Our second contribution aims to overcome the aforementioned undecidability results.
For that, we study approximations. More precisely, we focus on inferring bounding functions (in general depending on the input) on the expected values of classical program  variables upon termination.
The decision problem has thus been altered to an optimization problem.
Further, we restrict the set of potential bounding functions. 
As a suitable class of functions, we consider polynomials over the real-closed field of the algebraic numbers. 
The restriction to algebraic numbers guarantees that comparison operations between real numbers remain decidable. On the other hand for any real closed field, quantifier elimination for formulas over polynomials is decidable, that is, there
exists a double-exponential algorithm computing a quantifier-free formula equivalent to the original formula~\cite{HRS90}.
This recasting of the problem and restriction of the solution space suffices to render the problem decidable. 
The inference algorithm established remains double-exponential (Theorem~\ref{thm:bound}). Thus of similar complexity as the underlying quantifier elimination procedure.

Finally, our last contribution (Section~\ref{sec:pa}) studies effective automation of the inference of upper bounds on the expected values of
program variables. To improve upon the double-exponential complexity, we further restrict the class of polynomials considered, that is, to degree-2 polynomials and sketch how techniques from optimization theory can be employed.
Several simple quantum algorithms such as program $\rus$ can be analyzed using this approach (Example~\ref{ex:rusfinal}).
This further reduction in expressivity allows the encoding of the problem in SMT and thus paves the way towards (full) automation.

\subsubsection{Related work.}
Predicate transformers~\cite{D76,K85}---on which our work is based---were introduced as a method for reasoning about the semantics of imperative programs. They have been adapted  to the probabilistic setting, leading to the notion of expectation transformer~\cite{MM05,GKM14}, which has been used to reason about expected values~\cite{KK17,AMS23}, runtimes~\cite{KKMO16,NCH18}, and costs \cite{AMS20,ABL21,NCH18}, and to the quantum paradigm, leading to the notion of quantum pre-expectation transformer~\cite{OD20,LZY22,AMPPZ:LICS:22}. 

The problem of studying the difficulty of analyzing quantitative program properties has been deeply studied in the classical setting.
To mention a few, ~\cite{EGZ09} and~\cite{SchnablS11} study termination properties and runtime/derivational properties of first-order programs, respectively. Further, in~\cite{KaminskiKatoen} completeness results for various quantitative properties of (pure) probabilistic programs have been established.
The inference problem of expectation transformers, i.e., establishing nice implementation that automate the search for pre-expectations, has been studied intensively. Examples of successful implementation are presented in~\cite{NCH18,AMS20,AMS23}. 
Up to now, however, no practical, feasible studies have been carried out on quantum languages.
Among the techniques using quantum expectation transformers, we believe~\cite{AMPPZ:LICS:22} to be the most amenable to automation.
Indeed, by lifting \emph{upper invariants} of~\cite{KKMO16} to the quantum setting, it enables approximate reasoning and eliminate the need to reason about fixpoints or limits, stemming from the semantics of loops.

\section{Quantum programming language}
\label{s:l}

In this section, we introduce the syntax and operational semantics of the considered mixed-quantum imperative programming language.
\subsection{Syntax}
We make use of three basic datatypes $\Bool$, $\Var$ and $\Qubits$ for Boolean, numbers (non-negative integers), and qubit data, respectively. 
Let $\K$ be an arbitrary classical type in $\{\Bool,\Var\}$. 
Each program variable comes with a fixed datatype and  can be optionally annotated by its type as a superscript.
In what follows, we will use $\x,\x',\y,\ldots$ to denote classical variables of type $\K$ and $\q,\q',\ldots$ to denote quantum variables of type $\Qubits$.
A program, denoted $\qmd$, is simply a statement; see Figure~\ref{fig:synt}.
Program statements are either classical assignments, conditionals, sequences, loops, \emph{quantum assignments} $\qs^\Qubits <* \ope$, or \emph{measurements} $\x^\Bool <- \MEAS{\q^\Qubits}$.
A quantum assignment consists in the application of a quantum unitary gate $\ope$ of arity $ar(\ope)$ to a sequence of qubits $\qs \triangleq \q_1,\ldots,\q_{ar(\ope)}$ of type $\Qubits$. 
As we will see in the semantics section, a unitary matrix $U$ will be associated with each quantum gate $\ope$.
A measurement performs a single qubit measurement of $\q$ in the computational basis: the outcome is a Boolean value and the quantum state evolves accordingly. For a given syntactic construct $t$, let $\Bool(t)$ (respectively $\Var(t)$, $\Qubits(t)$) be the set of Boolean (respectively, number, qubit) variables in $t$.

\begin{figure*}[t]
\hrulefill
\\[10pt]
$ \begin{array}{rlll}
    \NExp &  \ni \nexp, \nexp_1, \nexp_2
    & \rgl &   \x^\Var \mid n \in  \mathbb{N} \mid
             \nexp_1 + \nexp_2   \mid  \nexp_1 - \nexp_2  \mid \nexp_1 \times \nexp_2 \\
    \BExp &  \ni \bexp, \bexp_1, \bexp_2
    & \rgl &  \x^\Bool   \mid \true \mid \false \mid
             \nexp_1 = \nexp_2  \mid  \nexp_1 < \nexp_2 \mid \neg \bexp \mid \bexp_1 \wedge \bexp_2 \mid \bexp_1 \vee \bexp_2 \\
    \Exp &  \ni \e, \e_1, \e_2
    & \rgl &  \nexp \mid \bexp \\
    \Cmds &  \ni  \cmd, \cmd_1, \cmd_2
    & \rgl &  \SKIP \mid  \x^\K <- \e^\K  \mid  \cmd_1; \cmd_2\mid\   \IF \bexp^\Bool \THEN \cmd_1  \ELSE \cmd_2   \\
          & & &\mid\ \WHILE \bexp^\Bool \DO \cmd
                \mid  \qs^\Qubits  <* \ope \mid   \x^\Bool <- \MEAS{\q^\Qubits}
  \end{array}$

 \hrulefill
 \vspace{-2mm}
 \caption{Syntax of quantum programs.}
\label{fig:synt}
\end{figure*}

\begin{example}
  We consider the program of Figure~\ref{fig:cointoss}, adapted from~\cite{AMPPZ:LICS:22}, as a simple leading example.
  Let $H$ be the unitary operator computing the Hadamard gate. This program simulates coin tossing by repeatedly measuring the qubit $\q$ in the loop body $\cmd$, given as initial input parameter until the measurement outcome $\false$ occurs.

  The probability to terminate within $n$ steps depends on the initial state $\rho=\smx{\alpha & \beta \\ \gamma & \delta}$  (a density matrix in $\Comp^{2\times 2}$, which implies $\alpha +\delta =1$ and $\gamma = \bar \beta$) of the qubit $\q$. Variable $\iv$ is increased by one at each iteration and, hence, when the program terminates, $\iv$ stores as final value the number of loop iterations performed. The overall probability of termination is $1$. The mean value of variable $\iv$, that is, the expected number of loop iterations, depends on the program input (hence, on the initital quantum state). 
The expected number of iterations  is thus given as
\begin{equation*}
F(\smx{\alpha & \beta \\ \bar \beta & \delta}) = p_0\times 1 + \sum_{i=1}^\infty \frac{p_1}{2^i}(i+1)=p_0+p_1+2p_1=1+(\alpha-\beta -\bar \beta+\delta)= 2-(\beta +\bar \beta) \tkom
\end{equation*}
where $p_0 = \frac{\alpha+\beta +\bar\beta+\delta}{2} = \frac{1+\beta + \bar\beta}{2}  = 1 - p_1$ is the probability to measure $\ket{0}$ on the first iteration of the loop. 
For a qubit initialized in state $\ket{\phi} = \sqrt{\sfrac{1}{3}}\ket{0}+\sqrt{\sfrac{2}{3}}\ket{1}$, the corresponding density matrix is $\rho_{\ket{\phi}} = \ketbra{\phi}{\phi}= \smx{\sfrac{1}{3}& \sfrac{\sqrt{2}}{3} \\ \sfrac{\sqrt{2}}{3}  & \sfrac{2}{3}}$ and hence the expected number of loop iterations is $F(\rho_{\ket{\phi}} )=2-2\sfrac{\sqrt{2}}{3}.$ It will be simply $2$ in the case of an initialization in the computational basis $\ket{\phi}=\ket{0}$ or $\ket{\phi}=\ket{1}$.
\end{example}

Notice that the language can be extended to deal with statements initializing qubits in the basis states. For example, the following syntactic sugar can be used $\q^\Qubits  <- \ket{0}\ \triangleq \  \x <- \MEAS{\q} ; \IF \x \THEN \q <* \oper{X}  \ELSE \SKIP,$ with $\oper{X}$ being the Pauli $X$ gate and for some fresh variable $\x$ of type $\Bool$.

\begin{figure*}[t]
\hrulefill
\\[5pt]
  \centering
  \begin{minipage}[t]{0.69\textwidth}
    $
    \coin \quad\triangleq\quad
    \pw{
      \x^\Bool <- \true ;\\
      \iv^{\Var} <- 0;\\
      \tikzrom{cq-mark1}
      \WHILE \x \DO {
        \iv <- \iv + 1;\tikzrom{cq-mar1}\\
        \q^\Qubits <* H;\hspace{3mm}\tikzrom{cq-mar3}\hspace{17mm}\tikzrom{cq-mark3}\\
        \x <- \MEAS{\q}\tikzrom{cq-mar2}
      }\tikzrom{cq-mark2}

    }
    $
\AddNotecopy{cq-mar1}{cq-mar2}{cq-mar3}{$\triangleq \cmd$}
\end{minipage}
\
\begin{minipage}[t]{0.29\textwidth}
$$\text{with }H=\frac{1}{\sqrt{2}}
 \begin{pmatrix}
  1& 1 \\
  1 & -1
   \end{pmatrix}$$
 \end{minipage}
\\[5pt]
\hrulefill
\vspace{-2mm}
\caption{Quantum Coin tossing}
\label{fig:cointoss}
\end{figure*}

\subsection{Operational semantics}\label{sec:os}

Following~\cite{AMPPZ:LICS:22}, we model the dynamics of our language as a probabilistic abstract reduction system (see~\cite{BG:RTA:05}), a transition system where reduction is determined as a relation over probability distributions.

\subsubsection{Probabilistic abstract reduction systems.}
Given a subset $\mathbb{K}$ of  $\mathbb{R}$, let $\mathbb{K}^+$ be the set of non-negative numbers in $\mathbb{K}$, i.e., $\mathbb{K}^+ \triangleq \mathbb{K} \cap \{x \mid x \geq 0 \}$ and let $\mathbb{K}^\infty$ be defined by $\mathbb{K}^\infty \triangleq \mathbb{K}\cup\{\infty\}$. 

A discrete (sub)\emph{distribution} $\delta$ over a set $A$ is a function $\delta : A \to [0,1] $ with countable support that maps an element $a$ of $A$ to a probability $\delta(a)$ such that $|\delta| \triangleq \sum_{a \in \supp(\delta)} \delta(a) = 1$ (respectively $|\delta | \leq 1$). Any (sub)distribution $\delta$ can be written as $\lmulti \delta(a) : a \rmulti_{a \in \supp(\delta)}$. The set of subdistributions over $A$, denoted by $\Dists(A)$, is closed under denumerable convex combinations $\sum_i p_i \cdot \delta_i \triangleq \lambda a.\sum_i p_i \delta_i(a)$, with $p_i \in [0,1]$ and $\sum_{i} p_i \leq 1$.
Slightly simplifying standard notation, given $f : A \to \Rext$ and a subdistribution $\delta \in \Dists(A)$, we define~$\mathbb{E}_\delta(f)$, \emph{the expectation of $f$ on $\delta$}, by
$\mathbb{E}_\delta(f) \triangleq \Sigma_{a \in supp(\delta)}\delta(a)f(a)$. Note that $\mathbb{E}_\delta(f) \in \Rext$ is always defined,
since the co-domain of $f$ is restricted to non-negative reals.

Bournez and Garnier \cite{BG:RTA:05} introduced the notion of
\emph{Probabilistic Abstract Reduction System} (PARS) as a means to study reduction systems that evolve probabilistically.
A PARS $\toop{}$ on $A$ is a binary relation $\cdot \toop{} \cdot \subseteq A \times  \Dists(A)$. The intended meaning is that when $a \to \delta$, then $a$ reduces to $b \in \supp(\delta)$ with probability $\delta(b)$.
Here, we focus on \emph{deterministic} PARSs, i.e., PARSs $\toop{}$ with $a \toop{c_1} \delta_1$ and $a \toop{c_2} \delta_2$ implies $\delta_1 = \delta_2$.
An object $a \in A$ is called \emph{terminal} if there is no rule $a \toop{c} \delta$, which we write as $a \not\to$.

Every deterministic PARS $\toop{}$ over $A$ naturally lifts to a reduction relation $\tomulti{}$ over distributions so that $\delta \tomulti{} \varepsilon$, if the reduct distribution $\varepsilon$
is obtained from $\delta$ by replacing reducts in $\supp(\delta)$ according to the PARS $\toop{}$.
In fact, we define this lifting in terms of a
ternary relation $\cdot \tomulti{\cdot} \cdot \ \subseteq \Dists(A) \times \Rpos \times \Dists(A)$ on distributions, where in a step $\delta \tomulti{c} \varepsilon$ the \emph{weight} $c$
signifies the probability that a reduction has occurred.
This relation is defined according to the
following three rules.
\begin{center}
  \begin{prooftree}
    \hypo{a \not\to\phantom{\tomulti{c }}}
    \infer1
    {\lmulti 1: a \rmulti \tomulti{0} \lmulti 1: a \rmulti}
  \end{prooftree}
  \quad \quad
  \begin{prooftree}
    \hypo{a \stackrel{\phantom{1}}{\to} \delta }
    \infer1
    { \lmulti 1: a\rmulti \tomulti{1} \delta  }
  \end{prooftree}
  \quad \quad
  \begin{prooftree}
    \hypo{\delta_i \tomulti{c_i} \epsilon_i}
    \hypo{\sum_i p_i \leq 1}
    \infer2
    {\sum_{i} p_i \cdot \delta_i \tomulti{\sum_i p_i c_i} \sum_{i} p_i \cdot \epsilon_i}
  \end{prooftree}
\end{center}
We may sometimes use the $n$-fold ($n \geq 0$) composition of $\tomulti{\cdot}$, denoted $\tomulti{\cdot}^n$, given by
$\delta \tomulti{c}^n \epsilon$ if $\delta \tomulti{c_1} \cdots \tomulti{c_n} \epsilon$ and the weights satisfy $c = \sum_{i=1}^n c_i$.
Notice that since $\toop{}$ is deterministic, so is $\tomulti{c}$ in the sense that
$\delta \tomulti{c_1} \epsilon_1$ and $\delta \tomulti{c_2} \epsilon_2$ implies $c_{1} = c_{2}$
and $\epsilon_{1} = \epsilon_{2}$.
Thus, in particular, for every $a \in A$ there is
precisely one infinite reduction
\[
  \{ 1: a \} = \delta_{0}
  \tomulti{c_0} \delta_{1}
  \tomulti{c_1} \delta_{2}
  \tomulti{c_2} \delta_{3} \tomulti{} \cdots .
\]
For any $b \in A$, the probability $\delta_{i}(b)$ gives the probability that $a$
reduces to $b$ in $i$ steps. Note that when $b$ is terminal, this probability only
increases along reductions (i.e., $\delta_{i}(b) \leq \delta_{i+1}(b)$ for all $i$).
This justifies that we define the \emph{terminal distribution} of $a$ as the distribution
$\delta(b) \triangleq \lim_{i \to \infty} \delta_{i}(b)$. Note that $\delta(b)$ gives the probability that $a$ reaches
$b$ in an arbitrary (but finite) number of steps.
Since the weights $c_{i}$ indicate the probability that a step has been performed from $\delta_{i}$
to $\delta_{i+1}$, the infinite sum $\sum_{i=0}^{\infty} c_{i} \in \Rext$ gives the expected
number of reduction steps carried out, the \emph{expected derivation length of $a$}.

For a PARS $\to$, we denote by $\nf[\to] : A \to \Dists(A)$ the function associating
with each $a \in A$ its terminal distribution. The \emph{expected derivation length function}
$\edl[\to] : A \to \Rext$ associates each $a \in A$ to its expected derivation length.
The PARS $\to$ is \emph{almost surely terminating}~\cite{sharir1984verification} (\emph{a.s. terminating} for short) if $a \in A$ reduces to a terminal object $b\not\to$
with probability $1$, that is, if $|\nf[\to](a)| = 1$ for every $A$.
It is \emph{positive almost surely terminating}, if the expected derivation length is always finite,
that is, $\edl[\to](a) < \infty$ for all $a \in A$.

Apart from termination, we are interested also in questions related to functional correctness,
such as (i)~what is the probability that $a$ reaches a terminal $b$,
(ii)~what is the probability that $a$ reaches a terminal satisfying predicate $P$, and more generally,
(iii)~which value does $f : A \to \Rext$ take, in expectation, when fully reducing an object $a$.
In the literature~\cite{MM05}, one tool to answer all of these are given by \emph{weakest pre-expectation
transformers}, the natural generalisation of classical weakest pre-condition transformers to a quantitative,
probabilistic setting. We suite this notion to PARSs.
\begin{definition}[Weakest pre-expectation]
  The \emph{weakest pre-expectation} for a PARS $\to$ over $A$ is given by the function
  \begin{align*}
    & \wp[\to] : {(A \to \Rext)} \to {(A \to \Rext)} \\
    & \wp[\to] \triangleq \lambda f. \lambda a.\ \E{\nf[\to](a)}{f} \tpkt
  \end{align*}
\end{definition}

For $\mathbf{1}_{b}$ the indicator function evaluating to $1$ on argument $b$ and to $0$ otherwise,
and by seeing a predicate $P$ as $0,1$ valued function,
$\wp[\to] \app \mathbf{1}_{b} \app a$ answers question (i),
$\wp[\to] \app P \app a$ answers (ii), and generally $\wp[\to] \app F \app a$ answers question (iii).
Note also that a PARS is a.s. terminating iff $\wp[\to] \app (\lambda b.\, 1) \app a = 1$ for each $a \in A$.
On the other hand, positive a.s. termination cannot be expressed through an application of $\wp[\to]$.


\subsubsection{Quantum programs as PARSs.}

\newcommand{\cst}{s}
\newcommand{\cstt}{r}
\newcommand{\qst}{\rho}
\newcommand{\tcon}[1]{#1^\dagger}
\newcommand{\cfg}[3]{(#1,#2,#3)}
\newcommand{\cfgc}[2]{(#1,#2)}
\newcommand{\halt}{\downarrow}
\newcommand{\dist}[1]{\{#1\}}
\newcommand{\dirac}[1]{\{ 1: #1 \}}
\begin{figure*}[h]
\hrulefill

\centering
$
\begin{prooftree}
  \hypo{\strut}
  \infer1[(Skip)]{\cfg{\SKIP}{\cst}{\qst} \tooqw \dirac{\cfg{\halt}{\cst}{\qst}}}
\end{prooftree}
\quad
\quad
\begin{prooftree}
  \hypo{\strut}
  \infer1[(Exp)]{\cfg{\x <- \e}{\cst}{\qst} \tooqw \dirac{\cfg{\halt}{\cst[\x := \sem[\cst]{\e}]}{\qst}}}
\end{prooftree}
$
\\
$
\begin{prooftree}
  \hypo{\strut}
\infer1[(Op)]{\cfg{\qs <* \ope}{\cst}{\qst} \tooqw \dirac{\cfg{\halt}{\cst}{\Phi_{U_{\qs}}(\qst)}}}

\end{prooftree}
$
\\
$
\begin{prooftree}
  \hypo{\strut}
  \infer1[(Meas)]{
    \cfg{\x <- \MEAS{\q_{i}}}{\cst}{\qst}
    \tooqw
    \dist{
      tr(\meas{k}{i}\qst):
\cfg{\halt}{\cst[\x := k]}{m_{k,i}(\qst)}}_{k \in \{0,1\}}}
\end{prooftree}
$
\\[0.2cm]
$
\begin{prooftree}
  \hypo{\cfg{\cmd_1}{\cst}{\qst} \tooqw \dist{ p_i : \cfg{\cmd_{\downarrow}^i}{\cst^i}{\qst^i}}_{i \in I}}
  \infer1[(Seq)]{
    \cfg{\cmd_1;\cmd_2}{\cst}{\qst} \tooqw \dist{ p_i : \cfg{\cmd_{\downarrow}^i;\cmd_{2}}{\cst^i}{\qst^i}}_{i \in I}}
\end{prooftree}
$
\\[0.2cm]
$
\begin{prooftree}
  \hypo{\sem[\cst]{\bexp} \in \{0,1\}}
  \infer1[(Cond)]{
    \cfg{\IF \bexp \THEN \cmd_{1} \ELSE \cmd_{0}}{\cst}{\qst} \tooqw \dirac{\cfg{\cmd_{\sem[\cst]{\bexp}}}{\cst}{\qst}}}
  \end{prooftree}
$
\\[0.2cm]
$
\begin{prooftree}
  \hypo{\sem[\cst]{\bexp} =0}
  \infer1[(Wh$_0$)]{\cfg{\WHILE \bexp \DO \cmd}{\cst}{\qst} \tooqw \dirac{\cfg{\halt}{\cst}{\qst}}}
\end{prooftree}
$
\\[0.2cm]
$
\begin{prooftree}
\hypo{\sem[\cst]{\bexp} =1}
\infer1[(Wh$_1$)]{\cfg{\WHILE \bexp \DO \cmd}{\cst}{\qst} \tooqw \dirac{\cfg{\cmd ; \WHILE \bexp \DO \cmd}{\cst}{\qst}}}
\end{prooftree}
$
\\[10pt]
\hrulefill
\vspace{-2mm}
\caption{Operational semantics in terms of PARS.}
\label{fig:os}
\end{figure*}

We now endow quantum programs with an operational semantics defined in terms of a PARS. 
Given a totally ordered set of qubits $Q=\{\q_1,\ldots,\q_n\}$, let $\mathcal{H}_{Q}$ be the $2^{n}$-dimensional Hilbert space defined by $\mathcal{H}_{Q} \triangleq \otimes_{i=1}^n \mathcal{H}_{\q_i}$, with $\mathcal{H}_{\q} = \Comp^{2}$ being the vector space of computational basis $\{\ket{0},\ket{1}\}$ and $\otimes$ being the tensor product. 
With $\bra{k}$ we denote the transpose conjugate of $\ket{k}$, for $k \in \{0,1\}$.
Let $\matrixspace$ be the set of complex square matrices acting on the Hilbert space $\mathcal{H}_{Q}$, i.e., $\matrixspace=\Comp^{2^n \times 2^n}$. 
Given $M \in \matrixspace$, $M^\dagger$ denotes the transpose conjugate of $M$ and let $I_{2^n}$ be the identity matrix over $\matrixspace$. We will write $I$ when the dimension is clear from the context. 

Let $\densop \subsetneq \matrixspace$ be the set of all \emph{density operators} (or quantum states), i.e., positive semi-definite matrices of trace equal to $1$ on $\mathcal{H}_{Q}$.  Density operators  can be viewed as the mathematical representation of a (mixed) \emph{quantum state}. 
A \emph{unitary operator} $U$ is a matrix in $\matrixspace$ such that $UU^\dagger = U^\dagger U = I$.
A \emph{superoperator} $\Phi_U : \densop \to \densop $, an endomorphism over density operators, is attached to each unitary operator $U$  and defined by $\Phi_U \triangleq \lambda \rho.U\rho U^\dagger$.
By definition, $\Phi_U$ is a completely positive trace preserving linear map. Indeed, $tr(U\rho U^\dagger)=tr(\rho)$, by unitarity. Hence $U\rho U^\dagger$ is a density operator in $\densop$ for each $\rho \in \densop$.


Regarding measurements, for each $i$, $1 \leq i \leq card(Q)$, we define $\meas{k}{i} \in \matrixspace$, with $k \in \{0,1\}$, by $\meas{0}{i} \triangleq  I_{2^{i-1}} \otimes (\ket{0}\bra{0}) \otimes I_{2^{n-i}}$ and $\meas{1}{i} \triangleq I - \meas{0}{i}$. The measurement of the qubit $q_i$ (in the computational basis) of a density matrix $\rho \in \densop$, produces the classical outcome $k\in \{0,1\}$ with probability $tr(\meas{k}{i}\rho)$. The (normalized) quantum state, after the measurement, is  defined by
\[
  m_{k,i} (\rho) \triangleq 
  \begin{cases}\frac{\meas{k}{i}\rho\meas{k}{i}^\dagger}{tr(\meas{k}{i}\rho)}, &\text{ if }tr(\meas{k}{i}\rho) \neq 0,\\
    \frac{I}{2^n} & \text{ otherwise.}
  \end{cases}
\]
Note that for all $\rho \in \densop,\ m_{k,i}(\rho) \in \densop$,  as it holds that $tr(m_{k,i}(\rho))=1$. Indeed, $tr(\meas{k}{i}\rho\meas{k}{i}^\dagger)= tr(\meas{k}{i}^2\rho)=tr(\meas{k}{i}\rho)$, as $\meas{k}{i}$ is a projection.
Hence $m_{k,i}$ is a map in $ \densop \to \densop$.

Set $\sem{\Bool}\triangleq \{0,1\}$ and $\sem{\Var} \triangleq \mathbb{N}$.
The classical state is modeled as a (well-typed) \emph{store} $s$ of domain $dom(s)$ mapping each variable $\x$ of type $\mathcal{K}$ to a value in $\sem{\mathcal{K}}$. 
Let $\cst[\x^\mathcal{K} := k]$ with $k \in \sem{\mathcal{K}}$ be the store obtained from $s$ by updating the value assigned to $\x$ in the map $s$. 
Given a store $\cst$, let $\sem[s]{-}: \mathcal{K}\mathtt{Exp} \to \sem{\mathcal{K}}$ be the map associating to each expression $\e$ of type $\mathcal{K}$ and such that $\Bool(\e)\cup \Var(\e) \subseteq dom(\cst)$, a value in $\sem{\mathcal{K}}$, and defined in a standard way.
For example $\sem[\cst]{\x} \triangleq \cst(\x)$, $\sem[\cst]{n} \triangleq n$, $\sem[\cst]{\true} \triangleq 1$, $\sem[\cst]{\nexp_1-\nexp_2} \triangleq \max(0, \sem[\cst]{\nexp_1}-\sem[\cst]{\nexp_2})$, etc.
Let $\Store$ be the set of stores.

Let $\halt$ be a special symbol for termination.
A \emph{configuration} $\mu$, for statement $\cmd \in  \Cmds \cup \{ \halt \}$, store $\cst$, and a quantum state $\qst$, has the form $\ltuple \cmd, \cst,\qst \rtuple$.
Let $\conf$ be the set of configurations. A configuration $\cfg{\cmd}{\cst}{\qst}$ is well-formed with respect to the sets of variables $B$, $V$, and $Q$ if $\Bool(\cmd) \subseteq B$, $\Var(\cmd) \subseteq V$,  $\Qubits(\cmd) \subseteq Q$, $dom(\cst) = B \cup V$, and $\qst \in \densop$.
Throughout the paper, we only consider configurations that are well-formed with respect to the sets of variables of the program under consideration.

The operational semantics is described in Figure~\ref{fig:os} as a PARS $\tooqw$ over objects in $\conf$, where terminal objects are precisely the configurations of the shape $\cfg{\halt}{\cst}{\qst}$.
The (classical or quantum) state of a configuration can only be updated by the three rules (Exp), (Op), and (Meas). 
Rule (Exp) updates the classical store wrt the value of the evaluated expression. 
Rule (Op) updates the quantum state to a new quantum state $\Phi_{U_{\qs}}(\qst) = U_{\qs} \rho U_{\qs}^\dagger$, where $U_{\qs}$ is the unitary operator  in $\matrixspace$  computed by extending the quantum gate $\oper{U}$ to the entire set of qubits $Q$.
Rule (Meas) performs a measurement on qubit $\q_i$.
This rule returns a distribution of configurations corresponding to the two possible outcomes, $k=0$ and $k=1$, with their respective probabilities $tr(\meas{k}{i} \qst)$ and, in each case, updates the classical store and the quantum state accordingly.
In the particular case where $tr(\meas{k_0}{i} \qst)=0$ for some $k_0\in\{0,1\}$, $ \dist{tr(\meas{k}{i}\qst): \cfg{\halt}{\cst[\x := k]}{m_{k,i}(\qst)}}_{k \in \{0,1\}} = \dist{1 : \cfg{\halt}{\cst[\x := 1-k_0]}{m_{1-k_0,i}(\qst)} }$.
Rule~(Seq) governs the execution of a sequence of statements $\cmd_1;\cmd_2$, under the covenant that $\halt ; \cmd \triangleq  \cmd$, for each statement $\cmd$. 
The rule accounts for potential probabilistic behavior when $\cmd_1$ performs a measurement and it is otherwise standard.
All the other rules are standard.

In a configuration $\mu = \cfg{\cmd}{\cst}{\qst}$, the pair $\sigma \triangleq (\cst,\qst)$ is called a state.
Let $\State^\cmd$ be the set of states $\sigma,\tau,\ldots$ that are well-formed wrt statement $\cmd$. 
For simplicity, we will denote this set by $\State$ when $\cmd$ is clear from the context. 
To ease the presentation, we sometimes write $\ltuple \cmd, \sigma \rtuple$ for the configuration $\mu$.

Specifically,
we will be interested also in expected runtimes and expectation based reasoning. These
carry over from PARSs as expected.
In what follows, we call functions $f : \conf \to \Rext$ \emph{expectations}.
\begin{definition}
For a statement $\cmd$, a state $\sigma$, and an expectation $f : \State \to \Rext$, we overload the notions of expected derivation length and weakest pre-expectation by:
  \begin{align*}
    & \edl[\cmd] : \State \to \Rext
    && \qwp[\cmd] : {(\State \to \Rext)} \to {(\State \to \Rext)} \\
    & \edl[\cmd]\triangleq \lambda \sigma. \edl[\tooqw] \ltuple \cmd, \sigma \rtuple
    &&  \qwp[\cmd] \triangleq \lambda f.\lambda \sigma. \wp[\tooqw](f_{st})\cfgc{\cmd}{\sigma}
       \tkom
  \end{align*}
  where $f_{st} \cfgc{\cmd}{\tau} = f(\tau)$.
\end{definition}

\begin{example}\label{ex:cointoss-sem}
  Consider the program $\coin$ given Figure~\ref{fig:cointoss}.
  In the setting of the program $\coin$, $Q=\{\q\}$,
  $\meas{0}{1} = \smx{1 & 0 \\ 0 & 0}$ and $\meas{1}{1} = \smx{0 & 0 \\ 0 & 1}$.
  %
  On an initial state $\sigma = (\cst,\qst)$, the reduction starts deterministically
  as in the classical setting, performing the initialization $\x = \true$ and $\iv = 0$.
  From there, evaluation reaches the loop $\WHILE\x\DO\cmd$. At each loop iteration,
  the loop counter $\iv$ is incremented, and the Hadamard gate applied to the quantum variable $\q$.
  The loop guard is obtained through measuring $\q$.

  To see how this is reflected in the semantics, let us first look at an iteration of the loop.
  If $\x$ was set to false, that is $\x$ holds the value $0$,
  by rule~(Wh$_0$) the loop terminates within one step:
  \begin{equation}
    \lmulti 1 : \cfg{\WHILE\x\DO\cmd}{[\x \asgn 0, \iv \asgn i]}{\qst} \rmulti
    \toomqw{1} \dirac{\cfg{\halt}{[\x \asgn 0, \iv \asgn i]}{\qst}} \tpkt \label{e0}\tag{0}
  \end{equation}
  On the other hand, when $\x$ was previously set to true, the loop executes its body. Precisely,
  we have:
  \begin{align}
    \mparbox{3mm}{\lmulti 1\! :\! \cfg{\WHILE\x\DO\cmd}{[\x\asgn1, \iv\asgn i]}{\qst} \rmulti} \notag{}\\
    &\toomqw{1}\! \lmulti 1\!  :\! \cfg{\iv <- \iv{+}1; \q <- \oper{H}; \x <- \MEAS{\q} ; \WHILE\x\DO\cmd}{[\x \asgn 1,\iv \asgn i]}{\!\qst} \rmulti \label{e1}\\
    &\toomqw{1}\! \lmulti 1 : \cfg{\q <- \oper{H}; \x <-\MEAS{\q};\WHILE\x\DO\cmd}{[\x \asgn 1,\iv \asgn i+1]}{\qst} \rmulti \label{e2}\\
    &\toomqw{1}\! \lmulti 1 : \cfg{\x <-\MEAS{\q}; \WHILE\x\DO\cmd}{[\x \asgn k,\iv \asgn i+1]}{\Phi_H(\qst)}\rmulti \label{e3}\\
    &\toomqw{1}\! \lmulti p_k:  \cfg{\WHILE\x\DO\cmd}{[\x \asgn k, \iv \asgn i+1]}{\qst_{k})} \rmulti_{k \in \{0,1\}}, \label{e4}
      \tkom
  \end{align}
  where in the last step, the probability $p_k$ equals~$\tr(\meas{k}{1}\Phi_H(\qst))$, while the normalized quantum
  state $\qst_{k}$ is given as~$m_{k,1}(\Phi_H(\qst))$.
  The above reduction is obtained by applying the rules of Figure~\ref{fig:os}:
  rule (Wh$_1$) for reduction (\ref{e1});
  rules (Exp) and (Seq) for reduction (\ref{e2});
  rules (Op) and (Seq) for reduction (\ref{e3});
  and finally rules (Meas) and (Seq) for reduction (\ref{e4}).

For an arbitrary initial quantum state
$\qst = \smx{\alpha & \beta \\ \gamma & \delta} \in \densop$
(where $\alpha,\beta,\gamma,\delta \in \Comp$ and $tr(\qst)=\alpha+\delta=1$, $\gamma = \overline{\beta}$, etc.),
it follows that
\begin{align*}
  p_0 = tr(\meas{0}{1} H\qst H^\dagger)
  = tr(\smx{1 & 0 \\ 0 & 0}
       \frac{1}{2}
       \smx{\alpha+\beta+\gamma+\delta &\alpha-\beta+\gamma-\delta\\
            \alpha+\beta-\gamma-\delta & \alpha-\beta-\gamma+\delta})
                                         = \frac{1+\beta + \gamma}{2}
                                         \tkom
\end{align*}
and likewise, $p_1 = 1 - p_0 = \frac{1 - (\beta + \gamma)}{2}$.
Using
$\qst_{k} = \frac{\meas{k}{1}H\qst H^\dagger\meas{k}{1}^\dagger}{tr(\meas{k}{1} H\qst H^\dagger)}
= \frac{(\meas{k}{1}H)\qst (\meas{k}{1}H)^\dagger}{p_k}$,
\begin{align*}
  \qst_{0} &
   = \frac{
        \smx{\sfrac{1}{\sqrt{2}} & \sfrac{1}{\sqrt{2}} \\ 0 & 0}\!
        \smx{\alpha & \beta \\ \gamma & \delta}\!
         \smx{\sfrac{1}{\sqrt{2}} & 0 \\ \sfrac{1}{\sqrt{2}} & 0}
      }{p_0}
    = \smx{1 & 0 \\ 0 & 0}
  &
  \qst_{1} &
   = \frac{
        \smx{ 0 & 0\\ \sfrac{1}{\sqrt{2}} & \sfrac{1}{\sqrt{2}}}\!
        \smx{\alpha & \beta \\ \gamma & \delta}\!
         \smx{ 0 & \sfrac{1}{\sqrt{2}} \\ 0 & \sfrac{1}{\sqrt{2}}}
      }{p_1}
   = \smx{0 & 0 \\ 0 & 1} \tpkt
\end{align*}
Summarizing (\ref{e1})--(\ref{e4}) we thus get:
\begin{align*}
\mparbox{4cm}{\lmulti 1 : \cfg{\WHILE\x\DO\cmd}{[\x \asgn 1, \iv \asgn i]}{\smx{\alpha & \beta \\ \gamma & \delta}} \rmulti}\\
&  \toomqw{4}^{4}
   \lmulti
   \begin{array}[t]{@{}r@{}l@{}}
     p_0 & : \cfg{\WHILE\x\DO\cmd}{[\x \asgn 0, \iv \asgn i+1]}{\qst_0}, \\[1mm]
     p_1 & : \cfg{\WHILE\x\DO\cmd}{[\x \asgn 1, \iv \asgn i+1]}{\qst_1} \rmulti.
   \end{array} \tpkt
\end{align*}
Putting everything together, we have
\[
\begin{array}{r@{\,}c@{\,}l}
  \cfg{\coin}{\cst}{\smx{\alpha & \beta \\ \gamma & \delta}}
  & \toomqw{\!\quad 2 \quad\! }^{2} & \lmulti 1 : \cfg{\WHILE\x\DO\cmd}{[\x\asgn1,\iv\asgn0]}{\qst} \rmulti \\
  & \toomqw{\!\quad 4 \quad\! }^{4} &
    \lmulti
    \begin{array}[t]{@{}r@{}l@{}}
      p_0 & : \cfg{\WHILE\x\DO\cmd}{[\x \asgn 0, \iv \asgn 1]}{\qst_0}, \\[1mm]
      p_1 & : \cfg{\WHILE\x\DO\cmd}{[\x \asgn 1, \iv \asgn 1]}{\qst_1} \rmulti
    \end{array}
  \\[1mm]
    & \toomqw{p_0+4p_1}^{4} &
    \lmulti
    \begin{array}[t]{@{}r@{}l@{}}
      p_0 & : \underline{\cfg{\halt}{[\x \asgn 0, \iv \asgn 1]}{\qst_0}}, \\[1mm]
      \frac{p_1}{2} & : \cfg{\WHILE\x\DO\cmd}{[\x \asgn 0, \iv \asgn 2]}{\qst_0}, \\[1mm]
      \frac{p_1}{2} & : \cfg{\WHILE\x\DO\cmd}{[\x \asgn 1, \iv \asgn 2]}{\qst_1} \rmulti
    \end{array}
  \\
    & \toomqw{\frac{p_1}{2}+4\frac{p_1}{2}}^{4} &
    \lmulti
    \begin{array}[t]{@{}r@{}l@{}}
      p_0 & : \underline{\cfg{\halt}{[\x \asgn 0, \iv \asgn 1]}{\qst_0}}, \\[1mm]
      \frac{p_1}{2} & : \underline{\cfg{\halt}{[\x \asgn 0, \iv \asgn 2]}{\qst_0}}, \\[1mm]
      \frac{p_1}{4} & : \cfg{\WHILE\x\DO\cmd}{[\x \asgn 0, \iv \asgn 3]}{\qst_0}, \\[1mm]
      \frac{p_1}{4} & : \cfg{\WHILE\x\DO\cmd}{[\x \asgn 1, \iv \asgn 3]}{\qst_1} \rmulti
    \end{array}
  \\[1mm]
      & \toomqw{\frac{p_1}{4}+4\frac{p_1}{4}}^{4} & \quad \cdots ,
\end{array}
\]
where terminal configurations are underlined.
This reduction converges to the terminal distribution
\begin{align*}
  \nf[\coin](\cst,\qst) =
    \begin{array}[t]{@{}l@{}}
      \lmulti p_0 : \cfg{\halt}{[\x \asgn 0, \iv \asgn 1]}{\qst_0} \rmulti
      + \lmulti \frac{p_1}{2^{i}} : \cfg{\halt}{[\x \asgn 0, \iv \asgn i+1]}{\qst_0} \rmulti_{i\geq 1}
    \end{array}
\tkom
\end{align*}
with an expected derivation length of
  \[
    \edl[\coin](\cst,\smx{\alpha & \beta \\ \gamma & \delta})
    = 2 + 4 + (p_0+4p_1) + \sum_{i=1}^\infty \frac{5p_1}{2^i} = 7+8p_1 = 11 - 4(\beta + \gamma)
    .
  \]
  For the expectation $f : \State \to \Rext$  defined by $f(\cst,\rho)\triangleq s(\iv)$ we have
  \[
    \qwp[\coin]\app f \app (\cst,\smx{\alpha & \beta \\ \gamma & \delta}) = p_0\times 1 + \sum_{i=1}^\infty \frac{p_1}{2^i}(i+1)=p_0+p_1+2p_1= 2 - (\beta + \gamma),
  \]
  that is, the mean value held by $\iv$ holds after execution is $2 - (\beta + \gamma)$.
  The termination probability is
  \[
    \qwp[\coin] \app (\lambda \sigma. 1) \app (\cst,\smx{\alpha & \beta \\ \gamma & \delta}) = p_0\times 1 + \sum_{i=1}^\infty \frac{p_1}{2^i} \times 1=p_0+p_1=1,
  \]
  i.e., the program is almost surely terminating.
\end{example}

\section{Weakest pre-expectations and the arithmetical hierarchy}
\label{s:mapping}

In this section, we study the hardness of some natural quantitative problems for weakest pre-expectations and expected derivation length.

\subsection{Computability-aimed restrictions}

\newcommand{\Alg}{\overline{\Q}}
\newcommand{\FADists}{\Dists_{\RAlg^+}^{\text{fin}}}
\newcommand{\confCT}{\conf_{\mathtt{CT}}}
\newcommand{\CmdsCT}{\Cmds_{\mathtt{CT}}}
\newcommand{\StateCT}{\State_{\mathtt{CT}}}
\newcommand{\ExCT}{{\tt E}_{\mathtt{CT}}}

This subsection is devoted to putting some restrictions on programs and on the considered notion of expectation to overcome the above computability issues.

\paragraph{Algebraic numbers.} Towards that end, a solution is to target a subset of complex numbers, where simple operations like equality are decidable.
We consider the set $\Alg$  of algebraic numbers, i.e., complex numbers in $\mathbb{C}$ that are roots of a non-zero polynomial in $\mathbb Q[X]$.
Let $\RAlg \triangleq \Alg \cap \mathbb{R}$ be the real closed field of real algebraic numbers in $\mathbb{R}$.
The following inclusions trivially hold
\begin{inparaenum}[(i)]
  \item $\mathbb{N} \subseteq \mathbb{Q} \subseteq \RAlg \subseteq \mathbb{R} \subseteq \mathbb{C}$ and
  \item $\Alg \subseteq \mathbb{C}$.
\end{inparaenum}
It was proved in~\cite[Proposition 2.2]{HalHarHir05} that equality over $\Alg$ and inequality over $\RAlg$ are decidable using Cohn's representation~\cite{cohn2002further} and is well-known that the product and sum over $\Alg$ are computable in polynomial time.

The program semantics is then restricted to matrices and density operators over the algebraic numbers. Given a totally ordered set of qubits $Q=\{\q_1,\ldots,\q_n\}$, let $\tilde{\mathcal{H}}_{Q}$ be the  Hausdorff pre-Hilbert space $\Alg^{2^{n}}$ (i.e., the completeness requirement on Hilbert spaces is withdrawn) of $n$ qubits defined by $\tilde{\mathcal{H}}_{Q} \triangleq \otimes_{i=1}^n \tilde{\mathcal{H}}_{\q_i}$, with $\tilde{\mathcal{H}_{\q}} \triangleq \Alg^{2}$ being the vector space of computational basis $\{\ket{0},\ket{1}\}$ over the field $\Alg$.
Let $\matrixspacealg$ and $\densopalg$ be the set of matrices and density operators on $\tilde{\mathcal{H}}_{Q}$, respectively. 

\paragraph{Clifford+T gates.} For the program semantics to be defined on the space $\densopalg$, the considered quantum gates will be restricted to gates whose corresponding unitary operators are  in $\matrixspacealg$, i.e., have a matrix representation over the algebraic numbers . We consider a restriction to the Clifford+T gates: $\oper{I}$, $ \oper{X}$, $\oper{Y}$, $\oper{Z}$, $\oper{H}$, $\oper{S}$, $\oper{CNOT}$, and $\oper{T}$, whose unitary matrices are given below:
\[
  I\triangleq\smx{1 & 0 \\ 0 &1},
  \
  X\triangleq \smx{0 & 1 \\ 1 &0},
  \
  Y\triangleq \smx{0 & -\img \\ \img &0},
  \
  Z\triangleq \smx{1 & 0 \\ 0 &-1},
  \
  H\triangleq\frac{1}{\sqrt{2}} \smx{1 & 1 \\ 1 &-1 },
\]
\[
  S\triangleq\frac{1}{\sqrt{2}} \smx{1 & 0 \\ 0 &\img},
  \
  CNOT\triangleq
  \smx{
    1 & 0 & 0 &0 \\
    0 & 1 & 0 & 0\\
    0 & 0 & 0 & 1\\
    0 & 0 & 1 & 0},
  \
  T\triangleq \smx{1 & 0 \\ 0 & e^{\img\frac{\pi}{4}}}
  \tpkt
\]
The Clifford+T fragment is the set of unitary transformations generated by sequential (matrix multiplication)  and parallel (Kronecker product) compositions of  the gates $H$, $S$, $CNOT$, and $T$. 
This constitutes a reasonable restriction for unitary operators as  it is known to be the simplest approximately universal fragment of quantum mechanics~\cite{aaronson2004improved}. 

Moreover, the superoperator of any unitary operator of the Clifford+T fragment is an endomorphism over density operators in  $\densopalg$.
 \begin{lemma}
The Clifford+T fragment preserves $\densopalg$, i.e., there exist $Q$ and $\qs \in Q$ such that for each unitary operator $U$ of the Clifford+T fragment $\Phi_{U_{\qs}} \in \densopalg \to \densopalg$. 
\end{lemma}

Notice that, while a restriction to Clifford+T is reasonable in terms of quantum mechanics and universality, our result can be extended by adding any quantum gate preserving the above lemma. For example, the phase shift gate, defined by $
P_\varphi\triangleq \smx{1 & 0 \\ 0 &e^{\img\varphi}}
$,  preserves $\densopalg$ whenever $\varphi = r \pi$, for any $r \in \mathbb{Q}$.

%

Let $\CmdsCT$ be the set of statements restricted to quantum gates computing Clifford+T unitary operators (hence a subset of $\Cmds$), $\StateCT$ be the set of states whose quantum state is in $\densopalg$, and $\confCT$ be the set of well-formed configurations in $(\CmdsCT \cup \{\halt \}) \times \StateCT$.
Let $\StateCT^\cmd$ be the set of states in $\StateCT$ that are well-formed wrt statement $\cmd$.  Once again, by abuse of notation, we will denote this set by $\State$ when $\cmd$ is clear from the context.

A crucial observation is that $\confCT$ is closed under reduction, in the following sense.
Let $\FADists(A) \subseteq \Dists(A)$ be the set of \emph{finitely supported} sub-distributions $\delta$ with algebraic probabilities, i.e., $\delta(a) \in \RAlg^+$ for all $a \in A$.
\begin{restatable}{lemma}{toctclosed}\label{l:to-ct-closed}
  The set $\FADists(\confCT)$ is stable under reduction, more precisely,
  if $\delta \in \FADists(\confCT)$ and $\delta \toomqw{c} \varepsilon$, then $\varepsilon \in \FADists(\confCT)$ and $c \in \RAlg^{+}$.
\end{restatable}

\paragraph{Computable expectations.}
We also restrict the expectation codomain to algebraic numbers. Hence the considered expectations will be functions in $ \StateCT \to \RAlg^+$. On its own, this restriction is not sufficient as the set $\StateCT \to \RAlg^+$ is not countable. It implies that there exist expectations  in $ \StateCT \to \RAlg^+$ that are not computable functions. To solve this issue, we adopt the convention by considering only computable expectations in $\StateCT \to \RAlg^+$:
\[
  \ExCT \triangleq \{ f \ | \ f : \StateCT \to \RAlg^+,\ f \text{ computable} \} \tpkt
\]
An immediate consequence of Lemma~\ref{l:to-ct-closed} is that $\nf[\cmd](\sigma) \in \Dists(\confCT)$ for any $\cmd \in \CmdsCT$ and $\sigma \in \StateCT$.
In consequence, $\qwp[\cmd] \app f \app \sigma$ is well-defined for all $f \in \StateCT$.
This justifies that in our treatment below, we restrict expectations to the class $\ExCT$.
However, keep in mind that despite Lemma~\ref{l:to-ct-closed}, the subdistribution $\nf[\cmd](\sigma)$, obtained at the limit, does not fall within $\FADists(A)$.
It is neither finite nor are probabilities algebraic ($\RAlg^{+}$ is not complete).
In particular, in general~$\qwp[\cmd] \app f \app \sigma$ is a real number, rather than an algebraic one.

\subsection{Quantitative problems}
We now define formally the quantitative problems that we study.
\paragraph{Testing problems.}
Some natural quantitative problems related to weakest pre-expectations are to determine for a given program $\cmd$, a given  state $\sigma$, a given expectation $f$, and a given algebraic number $a$, whether the corresponding weakest pre-expectation $ \qwp[\cmd] \app f \app \sigma$ is smaller or equal than $a$. In this setting, it makes sense to consider any possible relation in the set $ \{<,\leq,=,\geq,>\} \subseteq \mathcal{P}(\RAlg \times \RAlg)$ as one could be interested in finding precise values, (strict) upper or lower bounds.

\begin{definition}

The \emph{testing problem} sets $\WP{\mathcal{R}} \subseteq \confCT \times \ExCT \times \RAlg^+$, for $\mathcal{R} \in  \{<,\leq,=,\geq,>\}$,  are  defined by:
 $$(\cmd,\sigma,f,a) \in \WP{\mathcal{R}}  \defiff   (\qwp[\cmd] \app f \app \sigma) \mathrel{\mathcal{R}} a.$$
\end{definition}
The consideration of both $\WP{\leq}$ and $\WP{>}$ may seem redundant, as $\WP{>}$ can be viewed as the complement of $\WP{\leq}$. However, it makes perfect sense to distinguish both properties, when considering the corresponding universal problems. I.e.\ properties for any possible input.

\paragraph{Finiteness problem.}
Another problem of interest consists in checking whether the weakest pre-expectations produces some finitary output.
\begin{definition}
The \emph{finiteness problem} set $\WP{\neq \infty} \subseteq \confCT \times \ExCT$ is  defined by:
 $$(\cmd,\sigma,f) \in \WP{\neq \infty}  \defiff   \qwp[\cmd] \app f \app \sigma < \infty.$$
\end{definition}

\paragraph{Termination problems.}
We also define two termination problems for almost sure termination and positive almost sure termination:

\begin{definition}
  The sets of (\emph{positive}) \emph{almost-sure terminating} configurations $\AST \subseteq \confCT$ ($\PAST \subseteq \confCT$) are defined by:
  \begin{align*}
    (\cmd,\sigma) \in \AST &\defiff  |\nf[\cmd](\sigma)|=1\\
  (\cmd,\sigma) \in \PAST &\defiff  \edl[\cmd](\sigma) < \infty \tpkt
\end{align*}
\end{definition}
It is well-known that $\PAST \subsetneq  \AST$, cf.~\cite{BG:RTA:05}.

\paragraph{Universal problems.}
Another kind of natural problems arises if one tries to check some properties for each possible program input (i.e., for each state $\sigma$). We can thus defined universal properties for each of the sets described previously.
\begin{definition}
  The sets of universal testing, finiteness and (positive) a.s. termination problems are defined by:
\end{definition}
\begin{align*}
    (\cmd,f,g) \in \UWP{\mathcal{R}} \subseteq \CmdsCT \times \ExCT^2  &\iff   \forall \sigma \in \StateCT, \ (\cmd,\sigma,f,g(\sigma)) \in \WP{\mathcal{R}},\\
   (\cmd,f) \in \UWP{\neq \infty} \subseteq \CmdsCT \times \ExCT
   & \iff   \forall \sigma \in \StateCT, \ (\cmd,\sigma,f) \in \WP{\neq \infty},\\
  \cmd \in \UAST \subseteq \CmdsCT
   &\iff  \forall \sigma \in \StateCT, \ (\cmd,\sigma) \in \AST,\\
  \cmd \in \UPAST \subseteq \CmdsCT
   & \iff  \forall \sigma  \in \StateCT, \ (\cmd,\sigma) \in \PAST \tpkt
\end{align*}

\begin{example}
  We have $\coin \in \UAST$ and $\coin \in \UPAST$, for the program $\coin$ of Figure~\ref{fig:cointoss}.
  Indeed, it was shown in Example~\ref{ex:cointoss-sem} that $\coin$ terminates with probability $1$ and a finite expected derivation length. This property holds for any input of the domain.
  In the same example, we have proven $(\coin,f) \in \WP{\neq\infty}$ for $f(\cst,\qst) = \cst(i)$. Indeed, we have shown the stronger property
  $(\coin,f,g) \in \WP{=}$, where $g(\cst,\smx{\alpha & \beta \\ \gamma & \delta}) = 2 - (\beta + \gamma)$.
\end{example}

\subsection{Completeness results in the arithmetical hierarchy}
In what follows, we place the introduced quantitative problems within
the \emph{arithmetical hierarchy}~\cite{odifreddi1992classical}.
The arithmetical hierarchy is a means to classify and relate \emph{undecidable} problems wrt.\ to their inherent difficulty,
measured in terms of the number of (unbounded) quantifier alternations needed to state the problem
as a formula in first-order arithmetic, based on a decidable (recursive) predicate.


\paragraph*{Reminder on the arithmetical hierarchy.}
Classes of the arithmetical hierarchy are defined inductively as follows: 
\begin{align*}
\Pi^0_0 = \Sigma_0^0&\triangleq \textsc{rec}, \quad \text{ \textsc{rec} being the class of decidable problems (recursive sets)} \\
\Pi^0_{n+1} &\triangleq \{ \psi \ | \ \exists \phi \in \Sigma^0_n,\ \forall \overline{x}. (\psi(\overline{x})  \iff  \forall \overline{y}.\phi(\overline{x},\overline{y}))\},\\
\Sigma^0_{n+1} &\triangleq \{ \psi \ | \ \exists \phi \in \Pi^0_n,\ \forall \overline{x}. (\psi(\overline{x})  \iff  \exists \overline{y}.\phi(\overline{x},\overline{y}))\} \tpkt
\end{align*}
For each $n$, $\Pi^0_n$ is the complement of $\Sigma^0_n$ (i.e., $\Pi^0_n =$ co-$\Sigma^0_n$, and vice versa) and it is a well-known result that $\Sigma^0_ 1$ and  $\Pi^0_1$ correspond to the classes $\textsc{re}$ of recursively enumerable (i.e., semi-decidable) problems and co-$\textsc{re}$ of co-recursively enumerable (i.e., co-semi-decidable) problems, respectively.
Given the sets $A \subseteq X$ and $B \subseteq Y$, we write $A \reduces B$ ($A$ is many-one reducible to $B$) if there exists a computable function $f: X \to Y$ such that $\forall x \in X,\ x \in A \iff f(x) \in B$.
Given a class $\textsc{c}$ of the arithmetical hierarchy and a set $A$, $A$ is $\textsc{c}$-\emph{hard} if  $\forall B \in \textsc{c},\ B \reduces A$.
A set $A$ is  $\textsc{c}$-\emph{complete} if $A \in \textsc{c}$ and $A$ is $\textsc{c}$-hard.
It is well-known that if a set $A$ is $\textsc{c}$-complete then its complement, noted co-$A$, is (co-$\textsc{c}$)-complete.

\newcolumntype{C}{>{$}w{l}{2cm}<{$}}
\newcolumntype{D}{>{$}c<{$}}

\newcommand{\mrk}{${}^{(\ddag)}$}
\begin{table}[t]
  \newcommand{\FWP}{\WP{\neq\infty}}
  \newcommand{\FUWP}{\UWP{\neq\infty}}

  \centering
  \begin{tabular}{>{\em}l @{\qquad\qquad}  CDl w{c}{5mm} CDl}
    \toprule
    \multicolumn{1}{c}{} & \multicolumn{3}{c}{Standard} && \multicolumn{3}{c}{Universal}\\
    \cmidrule{2-4} \cmidrule{6-8}
    \multicolumn{1}{c}{} & \multicolumn{1}{c}{Problem} & \multicolumn{2}{c}{Class} && \multicolumn{1}{c}{Problem} & \multicolumn{2}{c}{Class} \\
    \cmidrule(r){1-1}
    \cmidrule{2-8}
    Testing
     & \WP{>}    & \Sigma^0_1 &      && \UWP{>}    & \Pi^0_2 & \mrk \\
     & \WP{\geq} & \Pi^0_2    & \mrk && \UWP{\geq} & \Pi^0_2 & \mrk \\
     & \WP{=}    & \Pi^0_2    &      && \UWP{=}    & \Pi^0_2 & \mrk \\
     & \WP{\leq} & \Pi^0_1    & \mrk && \UWP{\leq} & \Pi^0_1 & \mrk \\
     & \WP{<}    & \Sigma^0_2 &      && \UWP{<}    & \Pi^0_3 & \mrk \\[2mm]
    Finiteness
     & \FWP      & \Sigma^0_2 &      && \FUWP      & \Pi^0_3 & \mrk \\[2mm]
    Termination
     & \AST      & \Pi^0_2    &      && \UAST      & \Pi^0_2\\
     & \PAST     & \Sigma^0_2 &      && \UPAST     & \Pi^0_3\\[1mm]
    \bottomrule
  \end{tabular}
  \vspace{2mm}
  \caption{Completeness results for quantitative problems in the arithmetical hierarchy.}\label{fig:mapping}
\end{table}

\paragraph*{Results.}
Table~\ref{fig:mapping} associates the quantum decision problems to the corresponding classes
in the arithmetical hierarchy for which we have proven them \emph{complete}, that is, we have proven membership and hardness for the corresponding class. Some of the results may seem surprising. For instance, the testing problem
$\WP{>}$, i.e., deciding $\qwp[\cmd] \app f \app \sigma > a$ within the Clifford+T fragment,
turns out to be recursive enumerable.
It is thus not more difficult than the \emph{halting problem} $\mathcal{H}$.\footnote{In our context the halting set $\mathcal{H}$ can be defined
as the class of classical programs and states $(\qmd,\sigma)$ for which $\qmd$ is halting on~$\sigma$.}
Remarkable, through the restriction to the Clifford+T fragment, corresponding problems are ranked within the arithmetical
hierarchy identical to their non-quantum counterparts (\cite{SchnablS11,KaminskiKatoen}).
This observation holds for all problems apart those marked with (\ddag) which, to the best of our knowledge,
have not been studied in a classical/probabilistic setting.

A crucial observation towards these results is that, restricting to the Clifford+T fragment, the weakest pre-expectation
of a program $\qmd$ can be \emph{approximated} through \emph{computable} transformers
$\qwp<n>[\cmd] : \ExCT \to \ExCT$
that limit execution of $\cmd$ to at most $n \in \N$ reduction steps.
That is,
$$
  \qwp<n>[\cmd] \app f \app \sigma \triangleq \E{\nf<n>[\cmd](\sigma)}{f} \tkom
$$
for $\nf<n>[\cmd](\sigma)$ 
the distribution of terminal configurations obtained within $n$ reduction steps, when evaluating $\cfgc{\cmd}{\sigma}$.
With regards to the above mentioned $\WP{>} \in \Sigma^0_1$ for instance, observe that:
\begin{align*}
  (\cmd,f,\sigma,a) \in \WP{>}
  & \iff \qwp[\cmd] \app f \app \sigma > a \\
  & \iff \lim_{i \to \infty} \qwp<n>[\cmd] \app f \app \sigma > a \\
  & \iff \exists n \in \N,\, \exists \delta \in \RAlg^{+} \setminus \{0\},\ \qwp<n>[\cmd] \app f \app \sigma \geq a + \delta \tpkt
\end{align*}
Crucially, the predicate $\qwp<n>[\cmd] \app f \app \sigma \geq a + \delta$ becomes computable. In essence, this is a consequence of
Lemma~\ref{l:to-ct-closed}: The $n$-th step normal form distribution $\nf<n>[\cmd](\sigma)$
is finite and computable, as $f$ is computable so is thus $\qwp<n>[\cmd] \app f \app \sigma$.
From here, the result follows now as equality on $\RAlg$ is decidable.
The proof of this,
as well as all completeness proofs listed in Table~\ref{fig:mapping}
can be found in the Appendix.
The following constitutes our first main result.
\begin{theorem}\label{thm:mapping}
  All completeness results in Table~\ref{fig:mapping} hold.
\end{theorem}

\section{Quantum expectation transformers}
\label{s:wp}
\newcommand{\qmap}{\chi}
\newcommand{\proba}[2]{p_{#1,#2}}
We now define a notion of \emph{quantum expectation transformer} as a means to compute symbolically the weakest pre-expectation of a program.
To this end, we first introduce some preliminary notations in order to lighten the presentation.
\paragraph*{Notations.}
  For any expression $\e$, $\sem{\e}$ is a shorthand notation for the function $\lambda \ltuple \cst, \qst \rtuple. \sem[\cst]{\e} \in \State \to \Rext$.
  We will also use $f[\x := \e]$ for the expectation $\lambda \ltuple \cst, \qst \rtuple.f\ltuple \cst[\x := \sem[\cst]{\e}], \qst \rtuple$.
Similarly, for a given map $\qmap : \densop \to  \densop$,  $f[\qmap] \triangleq \lambda (\cst,\qst). f(\cst, \qmap(\qst) )$.   We will also sometimes group such state modifications,   for instance, $f[\x := \e\text{; }\qmap]$ stands for $(f[\x := \e])[\qmap]$
and $f[\x := \e,\y:=\e']$ stands for $(f[\x := \e])[\y:=\e']$.

  For $p  \in \State \to [0,1]$ and $f,g \in \State \to \Rext$, $f \up{p} g$ denotes the
  function $\lambda \sigma. p(\sigma) \cdot f(\sigma) + (1-p(\sigma)) \cdot g(\sigma) \in \State \to \Rext$,
  similar we use $f \cdot g$ to denote $\lambda \sigma. f(\sigma) \cdot g(\sigma) \in \State \to \Rext$.
  Thus, for instance, $f[\x := \x + 1] +_{\sem{\x = 1}} f$ behaves like~$f$, except that $\x$ is first incremented when applied to states
  with classical variable $\x$ equal to $1$.
  In correspondence to the normalization of quantum state $m_{k,i}$, we define probabilities $\proba{k}{i} \triangleq \lambda \qst. tr(\meas{k}{i}\qst\meas{k}{i}^{\dagger})$.
  We overload this function from $\densop$ to $\State$ s.t. $\proba{k}{i}(\cst,\qst) = \proba{k}{i}(\qst)$.
  In this way, $f[\x := 0\text{; }m_{0,i}] \up{\proba{0}{i}} f[\x := 1\text{; }m_{1,i}]$
  computes precisely the expected value of $f$ on the distribution of states obtained by measuring the $i$-th qubit and
  assigning the outcome to classical variable $\x$.

  Finally, we denote by $\leq$ also the pointwise extension of the order from $\Rext$ to functions, that is,
  $f \leq g$ holds iff $\forall \sigma \in \State,\ f(\sigma) \leq g(\sigma)$.

\begin{figure*}[t]
\hrulefill
\begin{align*}
  \wpt{\SKIP}{f} & \triangleq f
  \\
  \wpt{\x <- \e}{f} & \triangleq f[\x := \e]
  \\
  \wpt{\cmd_1 ; \cmd_2}{f} &\triangleq \wpt{\cmd_1}{\wpt{\cmd_2}{f}}
  \\
  \wpt{\IF \bexp \THEN \cmd_1 \ELSE \cmd_2}{f} & \triangleq \wpt{\cmd_1}{f}\up{\sem{\bexp}} \wpt{\cmd_2}{f}
  \\
  \wpt{\WHILE \bexp \DO \cmd}{f} &\triangleq \lfp\left(\lambda F. \wpt{\cmd}{F} \up{\sem{\bexp}} f \right)
  \\
  \wpt{\qs <* \ope }{f} &\triangleq f[\Phi_{U_{\qs}}]
  \\
  \wpt{\x <- \MEAS{\q_i}}{f}
  & \triangleq f[\x := 0\text{; }m_{0,i}] \up{\proba{0}{i}} f[\x := 1\text{; }m_{1,i}] \tpkt
\end{align*}
\hrulefill
\vspace{-2mm}
\caption{Quantum expectation Transformer $\wpt{\cdot}{\cdot}$}
\label{fig:qwpt}
\end{figure*}

\begin{definition}[Quantum expectation transformer]
The \emph{quantum expectation transformer} consists in a program semantics mapping expectations to expectations in a continuation passing style
 \[
   \wpt{\cdot}{\cdot} : \Cmds \to (\State \to \Rext) \to (\State \to \Rext)
 \]
 and is defined inductively on statements in Figure~\ref{fig:qwpt}.
\end{definition}
This transformer corresponds to the notion of \emph{expected value transformer} of~\cite{AMPPZ:LICS:22} on the Kegelspitze $\CSd = (\Rext,+_{\mathtt{f}})$, with $+_{\mathtt{f}}$ being the forgetful addition.
In the case of loops, the least fixed point $\lfp$ is defined with respect to the pointwise
ordering on the function space $\State \to \Rext$. Equipped with this ordering, this space forms a $\omega$-CPO.
As the quantum transformer can be shown to be $\omega$-continuous, the fixed-point is always defined, cf.~\cite{Winskel:93}.

\begin{restatable}[Adequacy]{theorem}{tadequacy}\label{t:adequacy}
 The following holds:
 $$\forall \cmd \in \Cmds,\ \forall f : \State \to \Rext,\ \qwp[\cmd](f) = \wpt{\cmd}{f}.$$
\end{restatable}

Apart from continuity, the quantum expectation transformer satisfies several useful laws, see Figure~\ref{fig:idents}.
The \ref{idents:mono} Law permits us to reason modulo upper-bounds: actual expectations can be always substituted by upper-bounds.
It is in fact an immediate consequence from the \ref{idents:cont} Law,
which is defined for any $\omega$-chain $(f_i)_i$.
The \ref{idents:ui} Law constitutes a generalization of the notion of invariant stemming from Hoare calculus.
It is used to find closed-form upper-bounds $g$ to expectations $f$ of loops.
The pre-conditions state that $g$ should dominate $f$ on states where the loop would immediately exist, and otherwise, should remain invariant under iteration.
It is worth mentioning that this proof rule is not only sound, but also complete, in the sense that
any upper-bound satisfies the two constraints.
The following example illustrates the use of this rule on the running example.

\begin{figure*}[t]
\hrulefill
\begin{alignat*}{2}
 \quad
& \law[idents:cont]{continuity}                &  & \textstyle \wpt{\cmd}{\sup_i f_i} = \sup_i \wpt{\cmd}{f_i} \\
& \law[idents:mono]{monotonicity}              &  & f \leq g \Rightarrow \wpt{\cmd}{f} \leq \wpt{\cmd}{g}                                                                         \\
& \law[idents:ui]{up. invariance}             &~& (\sem{\neg \bexp}{\cdot}f \leq g \wedge \sem{\bexp}{\cdot}\wpt{\cmd}{g} \leq g) \Rightarrow \wpt{\WHILE \bexp \DO \cmd}{f} \leq g 
\end{alignat*}
\hrulefill
\vspace{-2mm}
\caption{Universal laws derivable for the quantum expectation transformer.}
\label{fig:idents}
\end{figure*}

\begin{example}\label{ex:cointossqwpt}
  Following up on Example~\ref{ex:cointoss-sem}, we
  over-approximate $\wpt{\coin}{f}$, for $f(\cst,\qst) = \cst(\iv)$
  the post-expectation measuring the classical variable $\iv$.

  To this end, observe that
  the function $g: \State \to \Rext$ is an upper-invariant (Figure~\ref{fig:idents}) to the while loop $\WHILE\x\DO\cmd$, given a post-expectation $f: \State \to \Rext$. Recall that the loop body $\cmd$ comprises $(\iv <- \iv{+}1; \q <* \oper{H}; \x <- \MEAS{\q})$. To fulfill
  the conditions of the \ref{idents:ui} Law the following inequalities have to be met
  \begin{align}
    \sem{\neg \x} \cdot f &\leq g  &
                                     \sem{\x} \cdot \wpt{\iv <- \iv{+}1;\q <* \oper{H}; \x <- \MEAS{\q}}{g} & \leq g \label{ee2}
                                                                                                              \tpkt
  \end{align}
  By unfolding the definition, we see
  \begin{align*}
    \mparbox{3mm}{\wpt{\iv <- \iv{+}1;\q <* \oper{H}; \x <- \MEAS{\q}}{g}}\\
    &= \wpt{\iv <- \iv{+}1}{\wpt{\q <* \oper{H}}{\wpt{\x <- \MEAS{\q}}{g}}} \\
    &= \wpt{\iv <- \iv{+}1}{\wpt{\q <* \oper{H}}{g[\x \asgn 0; m_{0,1}] +_{\proba{0}{1}} g[\x\asgn 1;m_{1,1}]}}\\
    &= \wpt{\iv <- \iv{+}1}{g[\x \asgn 0; m_{0,1}; \Phi_H] +_{\proba{0}{1}\cdot \Phi_{H}} g[\x\asgn 1;m_{1,1}; \Phi_H]} \\
    &= g[\x \asgn 0; m_{0,1}; \Phi_H;\iv \asgn \iv{+}1] +_{\proba{0}{1}\cdot \Phi_{H}} g[\x\asgn 1;m_{1,1}; \Phi_H;\iv \asgn \iv{+}1] \\
    &= \lambda (\cst,\qst). \sum_{k \in \{0,1\}} \proba{k}{1}(\Phi_H(\qst)) \cdot g(\cst[\x = k,\iv \asgn \iv{+}1],m_{k,1}(\Phi_H(\qst)))
      \tpkt
      \  
             \tpkt
  \end{align*}
  By using the identities computed already in Example~\ref{ex:cointoss-sem},
  we thus obtain
  \begin{equation}
    \label{ee3}
    \wpt{\cmd}{g}(s,\smx{\alpha & \beta \\ \gamma & \delta}) =
    \sum_{k \in \{0,1\}} p_k \cdot g(\cst[\x = k,\iv\asgn\iv{+}1],\qst_{k}) \tkom
  \end{equation}
  where, as in Example~\ref{ex:cointoss-sem},
  $p_{0} = \frac{1 + \beta + \gamma}{2}$,
  $p_{1} = \frac{1 - (\beta + \gamma)}{2}$
  $\qst_0 = \smx{1 & 0 \\ 0 & 0}$ and
  $\qst_1 = \smx{1 & 0 \\ 0 & 0}$

  We claim that
  $  g (\cst,\smx{
    \alpha& \beta \\
    \gamma& \delta}) \triangleq
  \cst(\iv) + \cst(\x) \cdot (2 - (\beta + \gamma))
  $
  is an upper-bound to the pre-expectation of the while loop wrt. to the post expectation $f$.
  To this end, we check \eqref{ee2}.
  The first inequality is trivially satisfied. Concerning the second, notice that by definition,
  \[
    g(\cst[\x = 0, \iv \asgn \iv{+}1],\smx{1 & 0 \\ 0 & 0}) = \cst(\iv) + 1
    \quad\text{and}\quad
    g(\cst[\x = 1, \iv \asgn \iv{+}1],\smx{0 & 0 \\ 0 & 1}) = \cst(\iv) + 3.
  \]
  By \eqref{ee3} we have
  \begin{align*}
    \wpt{\cmd}{g}(\cst,\smx{\alpha& \beta \\ \gamma& \delta})
  & = \frac{1+\beta+\gamma}{2} (\cst(\iv) + 1)
    + \frac{1-(\beta+\gamma)}{2} (\cst(\iv) + 3) \\
  & = (\cst(\iv) + 2) -(\beta+\gamma) = g(\cst, \smx{\alpha& \beta \\ \gamma& \delta}) \tkom
  \end{align*}
  from which now the second constraint follows by case analysis on the value of $\x$.
  Hence $\wpt{\WHILE\x\DO\cmd}{f} \leq g$ and, by monotonicity (Figure~\ref{fig:idents}),
  \begin{align*}
    \wpt{\coin}{f}(\cst,\smx{\alpha& \beta \\ \gamma& \delta})
    & \leq \wpt{\x <- \true ; \iv <- 0 }{g}(\cst,\smx{\alpha& \beta \\ \gamma& \delta}) \\
    &  = g([\x \asgn 1,\iv \asgn 0],\smx{\alpha& \beta \\ \gamma& \delta})
       = 2 - (\beta + \gamma) \tpkt
  \end{align*}
  Note that, in this case, the computed bound is exact.
\end{example}

One question of interest is to find the $\wpt{\cdot}{\cdot}$ of a given statement.
We obtain the following completeness results as a corollary of Theorems~\ref{t:uwp-eq-complete},~\ref{t:uwp-ub-complete}, and~\ref{t:adequacy} on the Clifford+T fragment.
\begin{corollary}\label{c:wpt-complete}
The following completeness results hold:
\begin{itemize}
\item $\{(\cmd,f,g) \in \CmdsCT \times \ExCT^2 \ | \  \forall \sigma, \ \wpt{\cmd}{f}(\sigma)\ = g(\sigma)\}$
is $\Pi^0_2$-complete.
\item $\{(\cmd,f,g) \in \CmdsCT \times \ExCT^{2} \ | \  \forall \sigma, \ \wpt{\cmd}{f}(\sigma)\ \leq g(\sigma)\}$
is $\Pi^0_1$-complete.
\end{itemize}
\end{corollary}
The same kind of result can be straightforwardly obtained for each of the quantitative problems defined in previous section. All the corresponding sets are undecidable: they are at best (co-)semi-decidable as illustrated by Figure~\ref{fig:mapping}. This motivates us for restricting the problem a bit further to find a class of functions for which the quantitative problems for $\wp[\cmd]\app f$ can be decided.

%

\newcommand{\CC}[1]{\greenleaf{$#1$}}

\section{Decidability of qet  inference over a real closed field}
\label{s:decidability}

Corollary \ref{c:wpt-complete} illustrates that it is not sufficient to relax the problem of finding the quantum expectation transformer of a given statement to upper bounds in order to make it decidable.
The undecidability of finding the quantum expectation transformer of a given program is due to two other issues:
\begin{itemize}
\item \emph{Issue 1:} the  computation of a fixed-point for $\wpt{\cdot}{\cdot}$ in the case of loops,
\item \emph{Issue 2:} the check for inequalities over the set $\ExCT$, whose first-order theory is not decidable.
\end{itemize}
This section is devoted to overcome these two issues by finding an expressive fragment on which the inference of an upper-bound of the quantum expectation transformer becomes decidable.

\subsection{Symbolic inference}
\newcommand{\TERM}{\mathsf{ETerm}}
\newcommand{\VS}{\mathsf{EVar}}

As a first step towards automated inference, we define a symbolic variant of the quantum expectation transformer in Figure~\ref{fig:approx}.
In the case of conditionals, loops, and measurements, we
 will use fresh variables for expectations; \CC{\text{side conditions}}  will guarantee that these variables
indeed denote (upper-bounds to) the corresponding expectations.
This means that the symbolic version yields correct results only when the expectations assigned  to these variables satisfy all the side conditions.
By solving the generated constraints, viz., by finding an interpretation of ascribed variables that
satisfy the imposed side-conditions, we effectively arrive at an inference procedure overcoming Issue 1.

\begin{figure*}[t]
  \hrulefill
  \[
    \begin{array}{r@{\ }c@{\ }l r}
      \qinfer{\SKIP}{F}                                     & \triangleq & F
                                                                                                                                              \\[1mm]
      \qinfer{\x <- \e}{F}                                  & \triangleq & F[\x := \e]
                                                                                                                                              \\[1mm]
      \qinfer{\cmd_1 ; \cmd_2}{F}                           & \triangleq & \qinfer{\cmd_1}{\qinfer{\cmd_2}{F}}
                                                                                                                                              \\[1mm]
     \qinfer{\begin{array}{@{}l@{}}\IF<\ell> \bexp\\ \THEN \cmd_1\\ \ELSE \cmd_2\end{array}}{F} & \triangleq & X_{\ell},  ~~ \text{with side-cond.}\ \CC{\begin{cases}\bexp \vdash \qinfer{\cmd_1}{F} \leq X_{\ell}   \\  \neg\bexp \vdash \qinfer{\cmd_2}{F} \leq X_{\ell} \end{cases}\hspace{-3mm}}\\[1mm]
        \qinfer{\WHILE<\ell> \bexp \DO \cmd}{F}               & \triangleq & X_{\ell},  ~~ \text{with side-cond.}\  \CC{\begin{cases} \bexp \vdash \qinfer{\cmd}{X_{\ell}} \leq X_{\ell}  \\ \neg\bexp \vdash F \leq X_{\ell}\end{cases}\hspace{-3mm}} \\[1mm]
      \qinfer{\qs <* \ope }{F}                              & \triangleq & F[\phi_{U_{\qs}}]
      \\[1mm]
      \qinfer{\x <- \MEAS<\ell>{\q_i}}{F} & \triangleq & X_{\ell}, ~~ \text{with side-cond.}\
        \CC{\begin{cases}
          {\proba{0}{i} = 0                      \vdash F[\x := 1\text{; }m_{1,i}] \leq X_{\ell}} \\
          {\proba{1}{i} = 0                      \vdash F[\x := 0\text{; }m_{0,i}] \leq X_{\ell}} \\
          {\begin{array}[h]{@{}l@{}}\proba{k}{i} \neq 0 \vdash 
                                                                         F[\x := 0\text{; }m_{0,i}]\\
                                                                         \quad \quad \up{\proba{0}{i}} F[\x := 1\text{; } m_{1,i}] \leq X_{\ell}
                                                                       \end{array} }\\
        \end{cases}\hspace{-3mm}} 
    \end{array}
  \]
  \vspace{-2mm}
  \hrulefill
  \vspace{-2mm}
\caption{Term representations of $\qinfer{\cdot}{\cdot}$ and their corresponding side-conditions.}
\label{fig:approx}
\end{figure*}

To formalize this approach, we associate  a unique label $\ell$ with each loop, conditional, and measurement, occurring in the considered program.
This follows standard practice in (static) program analysis, cf.~\cite{Nielson:2005}.
Notationally, we write $\WHILE<\ell> \bexp \DO \cmd$ / $\IF<\ell> \bexp \THEN \cmd_1 \ELSE \cmd_2$ / $\MEAS<\ell>{\q}$.
Such labels permit us to associate a unique expectation variable $X_\ell$ to each of these constructs.
Given a set of such expectation variables $\VS$, the set of terms $\TERM$, upon which the symbolic
quantum expectation transformer operates, is defined according to the following grammar:
\[
  \begin{array}{llll}
    \TERM & F,G & \rgl & X \mid F[\x := \e] \mid F[\qmap] \mid F \up{p} G,
  \end{array}
\]
where $X$ stand for arbitrary expectations variable in $\VS$. As stressed above, $X$ will be used
to denote certain expectations wrt. loops, conditionals, and measurements.
We have already introduced the notations $F[\x := \e]$ and $F[\qmap]$ to represent updates to the classical and
quantum state, respectively, of the expectation term $F$.
Here, $\qmap$ will always denote a finite composition of superoperators $\phi_U$ and measurements $m_{k,i}$.
By ensuring that normalization of quantum states $m_{k,i}(\qst)$ is never considered in the degenerate case of a zero-probability measurement $\proba{k}{i}(\qst)$,
it will thereby always be possible to write $\qmap$ as $\lambda \qst. \frac{{M}\qst {M}^\dagger}{tr(N\qst N^\dagger)}$,
for some $M \in \matrixspacealg$ in the Clifford+T fragment.
Finally, following the same reasoning, in the barycentric sum $F \up{p} G$ the
probability $p$ is a function in the quantum state, and will always be of general form
$\lambda \qst. \frac{tr(M \qst M^\dagger)}{tr(N \qst N^\dagger)}$, 
for some $M,N \in \matrixspacealg$ in the Clifford+T fragment.
Similar to before, we may use the notation $F[\x := \e;\qmap]$.

The symbolic variation of the expectation transformer can now be defined as
\[
  \qinfer{\cdot}{\cdot} : \Cmds \to \TERM \to \TERM,
\]
generating also a set of side-conditions of the shape $\Gamma \vdash  F \leq G$,
with the intended meaning that $G$ binds $F$ on all input states that satisfy the predicate $\Gamma$.
The full definition of $\mathsf{qinfer}$ is given in Figure~\ref{fig:approx}.
As already hinted, the side conditions ensure that introduced variables $X_{\ell}$ indeed yield an upper-bound
on the corresponding expectation, in the case of conditionals by case-analysis, and in the case of loops
via an application of the upper-invariant law from Figure~\ref{fig:idents}.
In the case of measurements, $m_{k,i}$ and $\proba{k}{i}$ are defined exactly as before.
Here, we single out the two cases where the probability of a measurement, either
$\proba{0}{i}(\qst) = tr(\meas{0}{i}\qst) = tr(\meas{0}{i}\qst\meas{0}{i}^\dagger)$
or $\proba{1}{i}(\qst) = 1 - \proba{0}{i}(\qst)$, is zero. This way, we avoid the case analysis
underlying the definition of $m_{k,i}$ and may,  wlog., assume that
it is indeed of the form $\lambda \qst. \frac{\meas{k}{i}\qst\meas{k}{i}^\dagger}{tr(\meas{k}{i}\qst\meas{k}{i}^\dagger)}$,
with non-zero trace $tr(\meas{k}{i}\qst\meas{k}{i}^\dagger)$.

  \newcommand{\MM}{\mathsf{m}}
  \newcommand{\WW}{\mathsf{w}}
  \newcommand{\XW}{X_\WW}
  \newcommand{\XM}{X_\MM}

\begin{example}\label{ex:qinfer}
  In correspondence to Example~\ref{ex:cointossqwpt}, let us consider the application
  of the inference procedure on the program $\coin$, wrt. to the
  post-expectation $f(\cst,\qst) = \cst(\iv)$.
  We label the loop and measurement with $\MM$ and $\WW$, respectively.

  Let $X$ denote the post-expectation $f$.
  Unfolding the definition, we see
  \begin{align*}
    \qinfer{\coin}{X}
    & = \qinfer{\x <- \true; \iv <- 0; \WHILE<\WW> \x \DO \cmd}{X}  \\
    &= \XW[\x \asgn 1; \iv \asgn 0]
      \tkom
  \end{align*}
  generating the two side-conditions
  $\CC{\x  \vdash \XM[\Phi_{H};\iv \asgn \iv{+}1] \leq \XW}$ and
  $\CC{\neg\x  \vdash X \leq \XW}$.
  The left-hand side of the first constraint is obtained from
  \begin{align*}
    \qinfer{\cmd}{\XW}
    & = \qinfer{\iv <- \iv{+}1}{\qinfer{\q <* H}{\qinfer{\MEAS<\MM>{\q}}{\XW}}} \\
    &  = \XM[\Phi_{H};\iv \asgn \iv{+}1] \tpkt
  \end{align*}
  Note that this expansion generates further constraints, this time on $\XM$ representing the measurement.
  Specifically, it yields the following constraints:
  \begin{align*}
    & \CC{\proba{1-k}{1} = 0 \vdash \XW[\x \asgn k; m_{k,1}] \leq \XM}, \qquad (\text{for  $k \in \{0,1\}$}), \ \\
    & \CC{\proba{0}{1} \neq 0 \neq \proba{1}{1} \vdash \XW[\x \asgn 0; m_{0,1}] \up{\proba{0}{1}} \XW[\x \asgn 1; m_{1,1}] \leq \XM}
    \tpkt
  \end{align*}

Using the analysis from Example~\ref{ex:cointossqwpt}, we interpret $\XW$ and $\XM$ as:
  \begin{align*}
    \alpha(\XW) & \triangleq \lambda (\cst,\smx{\alpha & \beta \\ \gamma & \delta}).\ \cst(\iv) + \cst(\x) (2 - (\beta + \gamma)),\\
    \alpha(\XM) & \triangleq \lambda (\cst,\smx{\alpha & \beta \\ \gamma & \delta}).\ \cst(\iv) + 2 - 2 \alpha
    \tpkt
  \end{align*}
  Furthermore, we interpret the input variable $X$ as $f$, i.e., $\alpha(X) \triangleq \lambda (\cst,\qst).\ \cst(\iv)$.
  Notice how $\alpha(\XW)$ just corresponds to the upper-invariant $g$ derived in Example~\ref{ex:cointossqwpt}.
  Using the assignment, it is now standard to check that it is a solution to the five constraints.
  For instance, considering states $\sigma=(\{\iv \asgn n, \x \asgn x\},\smx{\alpha & \beta \\ \gamma & \delta})$,
  the ultimate constraint amount to the implication
  \[
    \alpha \neq 0 \neq \delta \Rightarrow n +_{\alpha} (n + 2) \leq n + 2 - 2\alpha
    \tkom
  \]
  which trivially holds.
  Finally, recall $\qinfer{\coin}{X} = \XW[\x \asgn 1; \iv \asgn 0]$. This term is interpreted as
  $\lambda (\cst,\smx{\alpha & \beta \\ \gamma & \delta}).\ 2 - (\beta + \gamma)$, yielding the optimal bound
  computed in Example~\ref{ex:cointossqwpt}.
\end{example}

\begin{example}\label{ex:rusfinal}
  Re-consider program $\rus$ depicted in Figure~\ref{fig:rus}.
  We are interested in an upper-bound on the number of $T$-gates, counted by the program variable $\iv$.
  As before, we label the loop and measurement with $\MM$ and $\WW$, respectively.
  Let
  \[
    \cmd =  \overbrace{\q_2 <- \ket{0}; \ldots }^{\cmd_0} ; \x <- \MEAS<\MM>{\q_2} \tkom
  \]
  be the body of the while loop statement  (see Figure~\ref{fig:rus}).
  We proceed with the analysis backwards.
  By the rules of \Cref{fig:approx}
  it holds that $\qinfer{\cmd_0}{F} = F[\Phi; \iv\asgn\iv{+}2]$ for any $F$,
  where $\Phi$ gives the quantum state updates within $\cmd_0$.
  Unfolding definitions, we have
  \[
    \qinfer{\rus{}}{X} = \XW[\x\asgn 0;\iv\asgn 1]
    \text{ with }
    \CC{\x \vdash \XM[\Phi; \iv\asgn \iv{+}2] \leq \XW}
    ,
    \CC{\neg \x \vdash X \leq \XW} \tkom
  \]
  since, by the above observation,
  \[
    \qinfer{\cmd}{\XW} = \qinfer{\cmd_0}{\qinfer{\x <- \MEAS<\MM>{\q_2}}{\XW}}
    =\XM[\Phi;\iv\asgn\iv{+}2] \tkom
  \]
  subject to the following additional  constraints stemming from measurements:
  \begin{align*}
    & \CC{\proba{1-k}{2} = 0 \vdash \XW[\x \asgn k; m_{k,2}] \leq \XM}, \qquad (\text{for  $k \in \{0,1\}$}), \ \\
    & \CC{\proba{0}{2} \neq 0 \neq \proba{1}{2} \vdash \XW[\x \asgn 0; m_{0,2}] \up{\proba{0}{2}} \XW[\x \asgn 1; m_{1,2}] \leq \XM}
    \tpkt
  \end{align*}


  Taking
  $\alpha(X) \triangleq \lambda (\cst,\qst).\ \cst(\iv)$ and solving the constraints
  yields a constant upper bound of $\sfrac{8}{3}$ on the expected number of $T$-gates used by the program.
  This is due to the fact that the probability of the internal measurement is always~$\frac{3}{4}$.
  Note that this bound is tight.
\end{example}

The transformer $\mathsf{qinfer}$ can be linked to $\mathsf{qet}$ of course
only when variables $X_{\ell}$ are interpreted in a way that the side conditions generated by $\mathsf{infer}$
are met. 
To spell this out formally,
let $\alpha : \VS \to \ExCT$ be an \emph{assignment} of expectations to variables in $\VS$,
and let $\sem[\alpha]{F} : \ExCT$ denote the interpretation of $F \in \TERM$
under $\alpha$
defined in the natural way, e.g., $\sem[\alpha]{X_{\ell}} = \alpha(X_{\ell})$, $\sem[\alpha]{F[\chi]} = \sem[\alpha]{F}[\chi]$, etc.

We say that a constraint $\Gamma \vdash F \leq G$ is \emph{valid under $\alpha$}
if $\sem[\alpha]{F}(\sigma)\leq \sem[\alpha]{G}(\sigma)$ holds for all
states $\sigma \in \StateCT$ with $\Gamma(\sigma)$.
An assignment $\alpha$ is a \emph{solution} to a set of constraints $\mathcal{C}$, if
it makes every constraint in $\mathcal{C}$ valid.
Finally, we say $\alpha$ is a solution to $\qinfer{\cmd}{f}$ if it is a solution to the set of constraints
generated by $\qinfer{\cmd}{f}$. We have the following correspondence:

\begin{theorem}\label{th:infer-sound}
  For any $\alpha \in  \VS \to \ExCT$, if $\alpha$ is solution to $\qinfer{\cmd}{F} = G$, then it holds that
  $\wpt{\cmd}{\sem[\alpha]{F}} \leq \sem[\alpha]{G}$.
\end{theorem}

It is worth mentioning that the above procedure could have been defined without restriction to  the full space $\State \to \Rext$ of expectations. In this case, this symbolic approach is also complete, in the sense that
if $\wpt{\cmd}{f} = g$ then $\qinfer{\cmd}{X} = G$ for some $G$ such that the side-conditions have a solution $\alpha$, with $\alpha(X) = f$ and $\sem[\alpha]{G} = g$.
As our main focus is on decidability, however, we have made the choice to restrict ourself to the Clifford+T setting.

\subsection{Restriction to polynomials over the real closed field $\RAlg$}
\label{s:rcf}

\newcommand{\QINFER}{\textsc{Qinfer}}
We now turn our eyes towards constraint solving, addressing the remaining Issue 2 through
restricting the domain of expectations to \emph{polynomials over algebraic numbers}.
To be more precise, we consider the following problem.
\begin{definition}
  Let $E \subseteq \ExCT$ be a class of expectations.
  The \emph{inference problem} $\QINFER(E) \subseteq \CmdsCT \times E \times (\VS \to E)$ is given by
  \[
    (\cmd,f,\alpha) \in \QINFER(E) \iff \text{$\alpha[X := f]$ is solution to $\qinfer{\cmd}{X}$}
  \]
\end{definition}
In the above definition, $ (\cmd,f,\alpha) \in \QINFER(E)$ is satisfied if the statement $\cmd$ has solution $\alpha[X := f]$ wrt. the expectation $f$. Hence it can be seen as checking whether $f$ is a post-expectation for $\cmd$.
In particular, any solution $\alpha[X := f]$ constitutes an upper bound on the weakest pre-expectation of $f$ (see Theorem~\ref{th:infer-sound}). We will now see that $\QINFER(E)$ is decidable, for $E$ the set of \emph{(real algebraic) polynomial expectations}
of (arbitrary but fixed) degree $d$.
For states $\StateCT$ over $n$ classical variables $\y_{1},\dots,\y_n$ and $m$ qubits,
let $\RAlg^d[\StateCT]$ denotes the class of functions of \emph{polynomial expectations} of the form
\begin{equation}
  \label{eq:polyexp}
  \lambda (\{ \y_i := Y_i\}_{1 \leq i \leq n},(A_{j,k} + \img B_{j,k})_{1 \leq j,k \leq 2^m}).\ P \tkom
\end{equation}
where variables $Y_{i}$ refer  to the classical, where variables $A_{j,k}$ and $B_{j,k}$ refer to the real part and imaginary part, respectively, of each algebraic coefficient in the quantum state. Further, $P\in \RAlg[Y_1,\ldots,Y_n,A_{1,1},\ldots,A_{2^m,2^m},B_{1,1},\ldots,B_{2^m,2^m}]$ is a multivariate polynomial with coefficients in $\RAlg$.
The index $d$ refers to the (total) degree of the underlying polynomial $P$.
For instance,
\[
  \lambda (\{\x := X; \iv := I\},
  \smx{A_{1,1}+\img B_{1,1} & A_{1,2}+\img B_{1,2} \\ A_{2,1}+\img B_{2,1} & A_{2,2}+\img B_{2,2}}).
  \ I + X ( 2 - (A_{1,2} + A_{2,1}))
  \in \RAlg^2[\StateCT]
\]
One important remark here is that we allow for possibly negative polynomials whereas expectations only output positive real algebraic numbers. Consequently, some side conditions are put on the admissible coefficients $A_{j,k}$ and $B_{j,k}$ of the input density matrix to preserve this condition (the matrix is positive, has trace $1$, is hermitian). For example, $\sum_{i=1}^{2^m}A_{i,i}=1$, $\sum_{i=1}^{2^m}B_{i,i}=0$ (trace is $1$) and $\forall i,\ k,\ A_{i,k}=A_{k,i}$ and $B_{i,k}=-B_{k,i}$ (self-adjointness). One can easily check that the expectations defined in Example~\ref{ex:qinfer} are in $\RAlg^d[\StateCT]$, for $d \geq 1$.

The restriction to polynomials is made on purpose as quantifier elimination is decidable
in the theory of real closed fields, a well known result due to Tarski and Seidenberg. 
Recall that the theory of real closed fields is the first-order theory in which the primitive operations are multiplication, addition, the order relation $\leq$, and the constants $0$ and $1$. 
Consequently, the only numbers that can be defined are the real algebraic numbers.
Specifically, we will make use of the following result, quantifying the complexity of the quantifier elimination
decision procedure as a function exponential in number of variables, and double-exponential in the number of
quantifier alternations.

\begin{proposition}[{\cite[Theorem 6]{HRS90}}]\label{prop:QE}
  Let $\mathbf{A}$ be an integral ring over a real closed field $\mathbf{R}$.
  Let $\psi = Q_1 \vec{x}_1. Q_{2}\vec{x}_2. \cdots Q_{l} \vec{x}_l.\ \phi$
  be a formula in prenex-normal form, where $\forall k,\ Q_k \in \{\forall,\exists\}$, $Q_k \neq Q_{k+1}$, 
and $\phi$ is a quantifier-free formula over $i$ variables and $j$ atomic propositions of the shape $P \geq 0$, each $P$ being a polynomial of degree at most $d$ with coefficients in $\mathbf{A}$.
  There exists an algorithm computing a quantifier-free formula equivalent to $\psi$
  in time $O(|\psi|) \cdot (jd)^{i^{O(l)}}$.
\end{proposition}

As $\RAlg$ constitutes both an integral ring and a real closed field, the above theorem is in particular
applicable taking $\mathbf{A} = \mathbf{R} \triangleq \RAlg$. In the particular case where $\psi$ is a closed formula, the resulting quantifier-free formula is simply a Boolean combination of inequalities over constants from $\RAlg$.
Since we already observed that these can be decided in polynomial time,
the above proposition thus implies that validity of $\psi$ is decidable under the given time bound.

By restricting assignment $\alpha$ to polynomial expectations,
it becomes decidable to check that $\alpha$ is a solution to a given
constraint set $C$. Indeed, under such a polynomial assignment $\alpha$, a constraint $\Gamma \vdash F \leq G$
becomes expressible as a formula in the theory of real closed field $\RAlg$.
By letting $\alpha$ range over polynomial expectations
with undetermined coefficients, we can this way arrive at the main decidability result of this section.

\begin{restatable}{theorem}{thmbound}\label{thm:bound}
  For any degree $d \in \N$, $d\geq 1$, the problem $\QINFER(\RAlg^d[\StateCT])$ is decidable
  in time $2^{2^{dO(n)}}$,   where $n$ is the size of the considered program.
\end{restatable}

\subsection{Practical algorithm}\label{sec:pa}
Theorem~\ref{thm:bound} established a computable algorithm on the inference of upper bounds on weakest pre-expectation on
quantitative program properties of any given mixed classical-quantum program. Nevertheless, the complexity
of this algorithm --- double-exponential in the program size --- is forbiddingly high. In order to turn this procedure into a
practical algorithm, we have to tame this inherent complexity.
For this, significant further restrictions on the class of bounding functions are necessary. We propose to proceed as follows.
%
\begin{inparaenum}[(1)]
\item \emph{Bounding functions:} in~\eqref{eq:polyexp} we restricted the class of expectations to polynomials, which in turn yield a
  bound on the weakest pre-expectation. Based on an analysis of concrete examples considered in the literature (e.g.,~\cite{LZY22,AMPPZ:LICS:22}), this can
  be tightened further to degree 2~polynomials.
\item \emph{Approximate solutions}: Theorem~\ref{thm:bound} rests upon (the decidability) of quantifier elimination. Thus the constraints $C$
  induced through the symbolic inference of $\qinfer{\cmd}{X} = G$ ($G,X \in \TERM$) are solved exactly. Over-approximation, however, suffices, if we are only interested in soundness of the inference mechanism.
\end{inparaenum}

The restriction of the class of bounding functions is in essence a question of applicability of the automation, taking into account particular use-cases.
With respect to approximate solutions, we observe that the actual constraints $C$ considered have at most one quantifier alternation and admit a quantifier prenex of the form $\exists^\ast \forall^\ast$, that is, a sequence of existential quantifier follows by
a sequence of universal quantifiers. Roughly speaking the existential quantifiers refer to the inference of coefficients in the
bounding polynomials, while the universal quantifiers refer to program variables. It is well-known that universal quantification in optimization
problems can be turned into existential quantification, like Farka's lemma or generalizations thereof, cf.~\cite{Schrijver:1999,Handelman88}. (E.g.,~\cite{AMS20,LMZ:2022} for instances of this approach for the inference of expected program costs.)

Summarizing, the inference mechanism detailed in Section~\ref{s:rcf} can be over-approximated to generate purely existential
constraints. The latter can be effectively solved via SMT. 
We expect that (full) automation of the inference mechanism can captialize on these ideas. 
Working out the details and in particular implementation of an effective prototype is subject to future work.

\section{Conclusion and future work}
We have studied the complexity and inference of techniques for obtaining qualitative program properties.  One particular property of interest would be the cost of quantum programs, that is average time, average number of gates, mean value of a variable, etc. We show that these problems were undecidable in general by placing them in the arithmetic hierarchy and saw that inference could become decidable on a restricted fragment: quantum gates in Clifford+T and a function space with a decidable theory (polynomials of bounded degree over a real closed field). Further, we sketch how the latter can be transformed into an efficient synthesis method.

Many open questions remain. The studied notion of expectation transformer describes  \emph{local} properties of the quantum state, while it would be interesting to extend this technique to the \emph{global} state so as to study a mixed state in a quantum-only setting (without classical variables and stores). Another question of interest is to what extent a characterization of the quantum class $\ZBQP$, the class of problems computed by quantum programs in polynomial expected runtime, could be obtained using this tool.

%

\bibliography{references}

\newpage

\appendix

\section{Further PARS Preliminaries}\label{ss:a:pars}
Let $\toop{}$ be a PARS.
If $\delta \tomulti{c}^n \epsilon$ then $\delta$ reduces to $\epsilon$ in $n$ steps,
with an expected runtime of $c$.
Notice that since $\toop{}$ is deterministic, so is $\tomulti{c}^n$ in the sense that
$\delta \tomulti{c_1}^{n} \epsilon_1$ and $\delta \tomulti{c_2}^n \epsilon_2$ implies $c_{1} = c_{2}$
and $\epsilon_{1} = \epsilon_{2}$.
It thus justifies to define $\edl<n>[\to](a) \triangleq c$, when $\{1: a\} \tomulti{c}^n \delta$.
Likewise, we denote by $\nf<n>[\to](a) \triangleq \{ \delta(b) : b \mid b \not\to  \wedge\ \exists c \in \Rpos,\ \{1: a\} \tomulti{c}^n \delta \}$
the \emph{distribution of terminal objects} reachable within \emph{up-to} $n$ steps.

\begin{definition}[Terminal distribution]
  For a PARS $\to$ over $A$ and $a \in A$,
  the \emph{terminal distribution} of $a \in A$ is given by
  \[
    \nf[\to](a) \triangleq \sup_{n \in \N} \nf<n>[\to](a).
   \]
\end{definition}
Here, the supremum should be understood wrt. the pointwise ordering on distributions.
This supremum always exists, the probability of reaching a certain terminal object only increases during reductions, i.e., $\nf<n>[\to](a)$ is monotonic in $n$.
The (sub)distribution $\nf[\to](a)$ yields for each terminal object $b \not\to$ the probability that $a$ reduces to $b$ in a finite number of steps.
Specifically, the total weight $|\nf[\to](a)|$ corresponds to the probability that $a$ is terminating.

\begin{definition}[Expected derivation length]
  The \emph{expected derivation length} of a PARS $\to$ over $A$ is given by the function:
  \vspace{-2mm}
\begin{align*}
& \edl[\to] : A \to \Rext\\
& \edl[\to]\triangleq \lambda a. \sup_{n \in \N} (\edl<n>[\to] \app a).
 \end{align*}
\end{definition}

We define the following approximations of $\wp[\to]$:
\begin{align*}
  \wp<n>[\to](f) &\triangleq \lambda a.\E{\nf<n>[\to](a)}{f}
  &
    \wpeq<n>[\to](f)(\sigma) &\triangleq \begin{cases} \wp<n>[\to](f)(\sigma) & \text{if } n= 0, \\  \wp<n>[\to] \app f \app \sigma - \wp<n-1>[\to](f)(\sigma) &\text{if } n>0. \end{cases}
\end{align*}
          By definition, it holds that
          \begin{align*}
            \wp[\to](f)(\sigma)
            = \sup_{n} \wp<n>[\to](f)(\sigma)
            = \sup_{n} \sum_{i=0}^{n}\wpeq<i>[\to](f)(\sigma)
            = \sum_{i=0}^{\infty}\wpeq<i>[\to](f)(\sigma).
          \end{align*}
Notice also that the weight of a terminal distribution---thus the termination probability---can be expressed as a particular weakest pre-expectation.

Let $\underline{1} \triangleq \lambda a. 1$.
\begin{proposition}\label{p:sizeaswp}
$|\nf[\to] \app a |  = \wp[\to] \app \underline{1} \app a$.
\end{proposition}
\begin{proof}
  As $|\cdot|$ is an increasing and continuous function, the following holds:
  \begin{align*}
    |\nf[\to](a)| &= |\sup_{n \in \N} \nf<n>[\to](a)| = \sup_{n \in \N}|\nf<n>[\to](a)| = \sup_{n \in \N}\E{\nf<n>[\to](a)}{\underline{1}} =\sup_{n \in \N}\wp<n>[\to] \app \underline{1} \app a\\
                  &=\wp[\to] \app \underline{1} \app a.
  \end{align*}
\end{proof}

\section{Proofs from Section~\ref{s:mapping}}

\toctclosed*
\begin{proof}
  Let us first observe that if $\mu \tooqw \delta$ then $\delta \in \FADists(\confCT)$.
  By looking at the rules of Figure~\ref{fig:os}, it is immediate that right-hand sides $\delta$ are finitely supported.
  In addition, quantum states are only updated by rules (Op) and (Meas). In both cases, the resulting density operator is in $\mathfrak{D}(\tilde{\mathcal{H}_{Q}})$, as $\forall i,\  \forall k,\ \meas{k}{i} \in \mathfrak{D}(\tilde{\mathcal{H}_{Q}})$ and algebraic numbers are closed under product and finite sum.
  Following the same reasoning, in each rule of Figure~\ref{fig:os}, the probabilities are also algebraic numbers in $\RAlg^+$.
  Conclusively $\delta \in \FADists(\confCT)$.

  Now consider a step
  $\delta \toomqw{c} \varepsilon$, for $\delta \in  \FADists(\confCT)$. Then, since $\cdot \toomqw{\cdot}\cdot$ is deterministic, wlog $\delta$ is of the form $\{ p_{i} : \mu_{i} \}_{0 \leq i \leq n}$
  and $\varepsilon = \sum_{i=0}^{n} p_{i} \cdot \varepsilon_{i}$ for $\{1 : \mu_{i}\} \toomqw{\cdot} \varepsilon_{i}$, i.e., $\varepsilon_{i}$ is either $\{1 : \mu_{i}\}$ (when $\mu_{i}$ is terminal) or $\mu_{i} \tooqw \varepsilon_{i}$. In any case, by the previous observation the scaled distributions
  $\p_{i} \cdot \varepsilon_{i} \in \FADists(\confCT)$, and so is their finite summation $\varepsilon$.
\end{proof}

In the remainder, we establish the completeness results, summarized in Table~\ref{fig:mapping},
that is we prove Theorem~\ref{thm:mapping}. We start with some preliminaries.

\subsection{Preliminaries}
Let $\Cmdsk$ be the subset of $\Cmds$ that only contains classical statements (i.e., statements that do not contain any quantum feature).
This fragment of our language is clearly still a Turing-complete programming language.
It obviously holds that $\Cmdsk \subsetneq \CmdsCT$.
For the sake of the hardness proofs, we briefly introduce the following well-known problems and the corresponding class for which they are  complete (see~\cite{odifreddi1992classical}):
\begin{align*}
&\hspace{-3mm}\text{(Halting Problem)}\ \mathcal{H}: \ \Sigma_1^0\text{-complete}\\
&(\cmd,\cst) \in \mathcal{H} \subseteq  \Cmdsk \times   \Store \iff  \exists s' \in \Store, \ \exists n \in \mathbb{N}, \ (\cmd, \cst, \cdot) \toomqw{\cdot}^{n} (\downarrow,\cst',\cdot), \\
&\hspace{-3mm}\text{(Univeral Halting Problem)}\ \mathcal{UH}: \ \Pi_2^0\text{-complete}\\
&\cmd \in \mathcal{UH} \subseteq  \Cmdsk \iff  \forall s \in \Store, \ (\cmd, s) \in \mathcal{H}, \\
&\hspace{-3mm}\text{(Cofiniteness Problem)}\ \mathcal{COF}: \ \Sigma_3^0\text{-complete}\\
&\cmd \in \mathcal{COF} \subseteq  \Cmdsk \iff \compl{\{ s \ |\  (\cmd, s) \in \mathcal{H} \}} \text{ is cofinite}.
\end{align*}
It is well-known that if a set $A$ is $\textsc{c}$-complete then its complement, noted co-$A$, is (co-$\textsc{c}$)-complete and that if $A$ is $\textsc{c}$-complete and $A \reduces B$ then $B$ is $\textsc{c}$-hard.
Recall that a set is \emph{cofinite}, if its complement is finite.
In the above definition, we use a $\cdot$ for the quantities that are not relevant. E.g., the quantum state is useless for the halting problem with respect to classical statements.

The following syntactic sugar will be used throughout the paper to represent some statements in $\CmdsCT$:
\begin{align*}
\q^\Qubits <-  \ket{+} &\quad\triangleq\quad \q <-  \ket{0} ; \q  <* H \tkom\\
\x^\Bool <- \coin() &\quad\triangleq\quad \pw{\q' <- \ket{+};\x <- \MEAS{\q'}} \tkom\\
\x^\Var <- \geo() &\quad\triangleq\quad \x' <- \true;  \x = - 1;   \WHILE \x' \DO \{\x' <- \coin(); \x <- \x + 1\} \tkom
\end{align*}
for some fresh variables $\q'$ and $\x'$ of type $\Qubits$ and $\Bool$, respectively.
The two first notations initialize a qubit to the basis vector $\ket{0}$ of the standard basis and to the basis vector $\ket{+}$ of the orthogonal basis, respectively. The penultimate notation simulate a fair coin toss by performing a measurement on $\ket{+}$. The last notation amounts to sampling an integer from a geometric distribution, assigning probability $\frac{1}{2^{n+1}}$ to each integer $n \in \N$.

Given a classical statement $\qmd \in \Cmdsk$ and a store $\cst$, we will also define $\qmd(\cst)$ to be a statement in $\Cmdsk$ simulating $\qmd$ on variables initialized by values of $\cst$.
Given a computable enumeration of stores $(\cst_{i})_{i \in \N}$ and a variable $\nexp^\Var$, $\qmd(\cst_{\nexp})$ will be used as syntactic sugar for the statement
which, on input store $\cstt$, first initializes the classical state to $\cst_{\cstt(\nexp)}$, and then simulates~$\qmd$. The existence of such a program is implied by the Turing-completeness of $\Cmdsk$.

\subsection*{Termination Problems}

\begin{theorem}
  \label{t:ast-complete}\label{t:uast-complete}
  $\AST$ is $\Pi^0_2$-complete and the corresponding universal problem $\UAST$ is $\Pi^0_2$-complete.
\end{theorem}
\begin{proof}
\begin{proofcases}
  \proofcase{$\AST$}
  Observe that
  \begin{align*}
    (\cmd,\sigma) \in \AST
    & \iff |\nf[\cmd](\sigma)|=1 \\
    & \iff \sup_{n \in \N} |\nf<n>[\cmd](\sigma)| = 1 \\
    & \iff \forall \epsilon \in \Rpos,\ \exists n,\  |\nf<n>[\cmd](\sigma)| \geq 1 - \epsilon \\
    & \iff \forall \epsilon \in \RAlg^{+},\ \exists n,\  |\nf<n>[\cmd](\sigma)| \geq 1 - \epsilon
  \end{align*}
  Since $|\nf<n>[\cmd](\sigma)| \in \RAlg^{+}$ is computable, $\Pi^0_2$-membership of $\AST$ follows.
  Concerning $\Pi^0_2$-hardness,
we directly  obtain $\mathcal{UH} \reduces \AST$, via the reduction
  \begin{align*}
    r &: \Cmdsk \to  \confct\\
    r & (\qmd) \triangleq  ((\pw{i<- \geo(); \qmd(\textit{\cst}_{i})}), \sigma^{*}),
  \end{align*}
  for fresh variable $\pw{i}$ of type $\Var$ and arbitrary but fixed initial state $\sigma^{*}$.
  Observe how this program is almost surely terminating, precisely if $\qmd(\cst_{i})$ is terminating for all $i \in \N$, i.e., if $\qmd \in \mathcal{UH}$.

  \proofcase{$\UAST$}
  Since $\cmd \in \UAST \iff \forall \sigma \in \StateCT,\ (\cmd,\sigma) \in \AST$, by the first part of the theorem, $\UAST \in \Pi^0_2$.
  $\Pi^0_2$-hardness is proven identically to that of $\AST$, via the reduction $r(\qmd) \triangleq (\pw{i<- \geo(); \qmd(\textit{\cst}_{i})})$  for some fresh variable $\pw{i}$ of type $\Var$,  and showing that $\mathcal{UH} \reduces \UAST$. 
  \end{proofcases}
\end{proof}

\begin{theorem}
  \label{t:past-complete}\label{t:upast-complete}
  $\PAST$ is $\Sigma^0_2$-complete and the corresponding universal problem $\UPAST$ is $\Pi^0_3$-complete.
  \qed
\end{theorem}
\begin{proof}
  \begin{proofcases}
    
  \proofcase{$\PAST$}
  Membership is established as follows:
  \begin{align*}
    (\cmd,\sigma) \in \PAST
    & \iff  \edl[\cmd](\sigma) < \infty \\
    & \iff \sup_{n \in \N} \edl<n>[\cmd](\sigma) < \infty \\
    & \iff \exists c, \forall n,\  \edl<n>[\cmd](\sigma) < c.
  \end{align*}
  Since $\edl<n>[\cmd](\sigma) < c$ is computable, $\Sigma^0_2$-membership of $\PAST$ follows. Hardness can be shown using the reduction of \cite[Theorem~8]{KaminskiKatoen} and the statement $\x^\Bool <- \coin()$ to simulate a fair coin tossing.

  \proofcase{$\UPAST$}
  Since $\cmd \in \UPAST \iff \forall \sigma \in \StateCT,\ (\cmd,\sigma) \in \PAST$, by Theorem~\ref{t:past-complete}, $\UPAST \in \Pi^0_3$.
  Hardness can be shown using the reduction of \cite[Theorem~10]{KaminskiKatoen} and the statement $\x^\Bool <- \coin()$ to simulate a fair coin tossing.
  
  \end{proofcases}
\end{proof}

\subsection*{Testing problems}

To simplify notations, we extend the notion of approximated weakest pre-expectation (see Section~\ref{ss:a:pars}
transformers from PARSs to quantum programs:
$$\qwp<n>[\cmd](f)(\sigma) \triangleq \wp<n>[\tooqw](f_{st})\cfgc{\cmd}{\sigma} \quad \qwpeq<n>[\cmd](f)(\sigma) \triangleq \wpeq<n>[\tooqw](f_{st})\cfgc{\cmd}{\sigma}.$$
In a similar vein,
we set $\nf[\tooqw]\cfgc{\cmd}{\sigma} \triangleq \nf[\cmd](\sigma)$.



\begin{theorem}
  \label{t:wp-slb-complete}\label{c:wp-ub-complete}
  $\WP{>}$ is $\Sigma^0_1$-complete and $\WP{\leq}$ is $\Pi^0_1$-complete.
  \qed
\end{theorem}
\begin{proof}
  \begin{proofcases}
    \proofcase{$\WP{>}$}
  As $\qwp<n>[\cmd](f)$ is monotonically increasing, the following equivalences hold: 
  \begin{align*}
    (\cmd,\sigma,f,a) \in \WP{>}
    & \iff \qwp[\cmd](f)(\sigma) > a \\
    & \iff \sup_{n \in \mathbb{N}} \qwp<n>[\cmd](f)(\sigma) > a \\
    & \iff \exists k \in \mathbb{N}, \exists \delta > 0, \ \qwp<k>[\cmd](f)(\sigma) \geq a + \delta
  \end{align*}
 Consequently,   $\WP{>} \in \Sigma^0_1$.
  For $\Sigma^0_1$-hardness, we have $\mathcal{H} \reduces \WP{>}$
  via the reduction
  \begin{align*}
    r &: \Cmdsk \times \Store \to  \confct \times \ExCT \times \RAlg^+\\
    r &(\qmd,\cst) \triangleq  (\cmd_\qmd,\cst,\ketbra{+}, \underline{1},\sfrac{1}{2}) \tkom
\end{align*}
where
$\cmd_{\qmd} \triangleq (\pw{ x <- \MEAS{\q}; \IF x \THEN \qmd \ELSE \SKIP})$, for some fresh variable $\x$ of type $\Bool$.
Note that
  \[
    \qwp[\cmd_{\qmd}](\underline{1})(\cst,\ketbra{+}) = \frac{\qwp[\qmd](\underline{1})(\cst[\x := 1],\ketbra{1})}{2} + \frac{1}{2} =  \frac{\qwp[\qmd](\underline{1})(\cst,\ketbra{+})}{2} + \frac{1}{2} \tpkt
  \]
  After performing the measure, the program executes $\qmd$ with probability $\sfrac{1}{2}$. Note that $\ketbra{1}$ is the quantum state obtained from $\ketbra{+}$ just before executing $\qmd$. The second equality above holds as $\qmd \in \Cmdsk$ and $\x$ is fresh with respect to $\qmd$. Finally, $\qwp[\qmd](\cst,\ketbra{+})(\underline{1}) = 1$ iff $\qmd$ is terminating on $\cst$, as programs in $ \Cmdsk$ do not depend on the quantum state. Thus, it holds that:
  \[
    (\qmd,\cst) \in \mathcal{H}
    \iff \qwp[\cmd_{\qmd}](\underline{1})(\cst, \ketbra{+}) > \sfrac{1}{2}
    \iff r(\qmd,\cst) \in \WP{>}.
  \]

  \proofcase{$\WP{\leq}$}
  Since $\WP{\leq} =$ co-$\WP{>}$, the claim follows by the first part of the theorem.
  \end{proofcases}

\end{proof}

\begin{theorem}
  \label{t:wp-sub-complete}\label{c:wp-lb-complete}
  $\WP{<}$ is $\Sigma^0_2$-complete and $\WP{\geq}$ is $\Pi^0_2$-complete.
\end{theorem}
\begin{proof}
\begin{proofcases}
  \proofcase{$\WP{<}$}
    
  As $\qwp<n>[\cmd](f)$ is monotonically increasing, the following equivalences hold: 
  \begin{align*}
    (\cmd,\sigma,f,a) \in \WP{<}
    & \iff \qwp[\cmd](f)(\sigma) < a \\
    & \iff \sup_{n \in \mathbb{N}} \qwp<n>[\cmd](f)(\sigma) < a \\
    & \iff \exists b \in \RAlg^+-\{0\},\ \forall k \in \mathbb{N}, < \qwp<k>[\cmd](f)(\sigma) < a - b
  \end{align*}
 Consequently,   $\WP{<} \in \Sigma^0_2$.  To prove completeness, we show that co-$\mathcal{UH} \reduces \WP{<}$,
  via the reduction
  \begin{align*}
    r &: \Cmdsk  \to  \confct \times \ExCT \times \RAlg^+\\
    r & (\qmd) \triangleq  (\cmd,\sigma^*,f_{\qmd}, 1),
\end{align*}
  where $\cmd \triangleq (\pw{i <- \geo();k <- \geo()})$,   
  $\sigma^*$ is an arbitrary but fixed initial state and
  \[
    f_{\qmd}(\cst,\qst) \triangleq \qwpeq<\cst(\mathtt{k})>[\qmd](\underline{1})(\sigma_{\cst(\iv)}) \cdot 2^{\cst(\kv)+1},
  \]
for $(\sigma_{i})_{i\in \mathbb{N}}$ a sequence of states whose stores provide an enumeration of  classical states in $\Store$ for $\qmd$.
Being computed as a finite sum, $f_{\qmd}$ is a computable expectation in $\ExCT$. 
  Notice that after execution, the values held by $\iv$ and $\kv$ are independently distributed according to a geometric distribution,
  conclusively
  \[
    \qwp[\cmd](f_{\qmd})(\sigma^*)
    = \sum_{i=0}^{\infty} \frac{1}{2^{i+1}} \sum_{k=0}^{\infty} \frac{1}{2^{k+1}}\qwpeq<k>[\qmd](\underline{1})(\sigma_{i}) \cdot 2^{k+1}
    = \sum_{i=0}^{\infty}\sum_{k=0}^{\infty}\frac{\qwpeq<k>[\qmd](\underline{1})(\sigma_{i})}{2^{i+1}}
  \]
  Notice that $\qwpeq<k>[\qmd](\underline{1})(\sigma_{i}) = 1$ iff $\qmd$ terminates with precisely $k$ steps on input $\sigma_{i}$,
  and $0$ otherwise. Thus, if $\qmd$ is terminating on $\sigma_{i}$,
  then $\sum_{k=0}^{\infty}\frac{\qwpeq<k>[\qmd](\underline{1})(\sigma_{i})}{2^{i+1}} = \frac{1}{2^{i+1}}$,
  and $0$ otherwise. Using the identity $\sum_{i=0}^{\infty} \frac{1}{2^{i+1}} = 1$
  from this it is not difficult to establish that
  \[
    \qmd \not\in \mathcal{UH} \iff \qwp[\cmd](f_{\qmd})(\sigma^*) < 1
    \iff r(\qmd) \in \WP{<}.
  \]

  \proofcase{$\WP{\geq}$}
  Since $\WP{\geq} =$ co-$\WP{<}$, the claim follows by the first part of the theorem.
\end{proofcases}
\end{proof}

\begin{theorem}\label{t:wp-eq-complete}
  $\WP{=}$ is $\Pi^0_2$-complete.
\end{theorem}
\begin{proof}
  Membership follows from the second parts of Theorems~\ref{t:wp-slb-complete} and~\ref{t:wp-sub-complete}, respectively,
  since $\WP{=} = \WP{\leq} \cap \WP{\geq}$.
  By definition $(\cmd,\sigma) \in \AST$ iff $|\nf[\cmd](\sigma)|=1$. By Proposition~\ref{p:sizeaswp}, it holds that $|\nf[\cmd](\sigma)|  = \qwp[\cmd](\underline{1})(\sigma) $. Consequently, $(\cmd,\sigma) \in \AST$ iff $\qwp[\cmd](\underline{1})(\sigma) =1$ and the reduction $r(\cmd,\sigma) \triangleq (\cmd,\sigma,\underline{1},1)$
  shows $\AST \reduces \WP{=}$, hence $\WP{=} \in \Pi^0_2$ by Theorem~\ref{t:ast-complete}.
\end{proof}

\begin{theorem}
  \label{t:uwp-slb-complete}\label{t:uwp-lb-complete}
  $\UWP{>}$ is $\Pi^0_2$-complete and $\UWP{\geq}$ is $\Pi^0_2$-complete.
\end{theorem}
\begin{proof}
\begin{proofcases}
  \proofcase{$\UWP{>}$}
  
  That $\UWP{>} \in \Pi^0_2$ follows directly from Theorem \ref{t:wp-slb-complete}.
  Following the pattern of the proof of Theorem~\ref{t:wp-slb-complete},
  for hardness, we show $\mathcal{UH} \reduces \UWP{>}$
  via the reduction
  \begin{align*}
    r &: \Cmdsk \to  \CmdsCT \times \ExCT \times \ExCT \\
    r &(\qmd) \triangleq  (\cmd_\qmd, \underline{1},\underline{\sfrac{1}{2}}) \tkom
\end{align*}
  where
  $\cmd_{\qmd} \triangleq (\pw{\q <- \ket{+}; x<- \MEAS{\q}; \IF x \THEN \qmd \ELSE \SKIP})$, for some fresh variables $\x$ and $\q$ of type $\Bool$ and $\Qubits$, respectively.
  Note that
  \[
    \qwp[\cmd_{\qmd}](\underline{1})(\cst,\qst) = \frac{\qwp[\qmd](\underline{1})(\cst[\x := 1],\qst')}{2} + \frac{1}{2} =  \frac{\qwp[\qmd](\underline{1})(\cst,\qst)}{2} + \frac{1}{2}
  \]
where $\qst'$ is the quantum state obtained from $\qst$ just before executing $\qmd$, as after performing the measure, the program executes $\qmd$ with probability $\sfrac{1}{2}$.
The second equality above holds as $\qmd \in \Cmdsk$ and $\x$ is fresh with respect to $\qmd$.
 It holds that $\qwp[\qmd](\cst,\qst)(\underline{1}) = 1$ iff $\qmd$ is terminating on $\cst$, as programs in $ \Cmdsk$ do not depend on the quantum state.
  Thus
  \[
    \qmd \in \mathcal{UH}
    \iff \forall \sigma \in \StateCT,\ \qwp[\cmd_{\qmd}](\underline{1})(\sigma) > \sfrac{1}{2}
    \iff r(\qmd) \in \UWP{>}
  \]

  \proofcase{$\UWP{\geq}$}
  We have
  \[
    (\cmd,f,g) \in \UWP{\geq}
    \iff \forall \sigma \in \StateCT,\  (\cmd,\sigma,f,g(\sigma)) \in \WP{\geq}
  \]
  and conclusively $\UWP{\geq}$ is in $\Pi^{0}_{2}$ by Corollary~\ref{c:wp-lb-complete}.
  Exploiting that $\cmd \in \AST \iff \forall \sigma \in \StateCT,\ \qwp[\cmd](\underline{1})(\sigma) = 1$,
  the reduction $r$ defined by 
\begin{align*}
    r &: \Cmdsk \to  \CmdsCT \times \ExCT \times \ExCT \\
    r &(\qmd) \triangleq  (\qmd,\underline{1},\underline{1}),
\end{align*}
proves that $\UAST \reduces \UWP{\geq}$
  and, conclusively, $\UWP{\geq}$ is
   $\Pi^0_2$-hard, by Theorem~\ref{t:uast-complete}.

\end{proofcases}
\end{proof}

\begin{theorem}\label{t:uwp-eq-complete}
  $\UWP{=}$ is $\Pi^0_2$-complete.
\end{theorem}
\begin{proof}
  Membership is a direct consequence of Theorem~\ref{t:wp-eq-complete}.
  Hardness is provable by showing $\mathcal{UH} \reduces \UWP{=}$
  using the reduction
  \begin{align*}
    r &: \Cmdsk  \to  \CmdsCT \times \ExCT \times \ExCT\\
    r & (\qmd) \triangleq  (\cmd,f_\qmd, \underline{1}),
  \end{align*}
  where $\cmd$ and $f_\qmd$ are precisely taken as in the proof of Theorem~\ref{t:wp-sub-complete}.
  As reasoned there, independent of the initial state $\sigma^*$,
  \[ 
    \qwp[\cmd](f_\qmd)(\sigma^*)
    = \sum_{i=0}^{\infty}\sum_{k=0}^{\infty}\frac{\qwpeq<k>[\qmd](\underline{1})(\sigma_{i})}{2^{i+1}}
  \]
  for $(\sigma_{i})_{i\in \mathbb{N}}$ a sequence of states whose stores provide an enumeration of  classical states in $\Store$ for $\qmd$.
  Since $\sum_{k=0}^{\infty}\frac{\qwpeq<k>[\qmd](\underline{1})(\sigma_{i})}{2^{i+1}} = \frac{1}{2^{i+1}}$ iff
  $\qmd$ halts on $\sigma_{i}$ and $0$ otherwise,
  by the identity $\sum_{i=0}^{\infty} \frac{1}{2^{i+1}} = 1$  it follows that
  \[
    \qmd \in \mathcal{UH} \iff \forall \sigma \in \StateCT,\ \qwp[\cmd](f_\qmd)(\sigma^*) = 1 \iff r(\qmd) \in \UWP{=}
  \]
\end{proof}

\begin{theorem}\label{t:uwp-ub-complete}
  $\UWP{\leq}$ is $\Pi^{0}_{1}$-complete.
\end{theorem}
\begin{proof}
  We have
  \[
    (\cmd,f,g) \in \UWP{\leq}
    \iff \forall \sigma \in \StateCT,\ (\cmd,\sigma,f,g(\sigma)) \in \WP{\leq}
  \]
  and conclusively $\UWP{\leq}$ is in $\Pi^{0}_{1}$ by the second part of Theorem~\ref{c:wp-ub-complete}.
  To conclude, it is sufficient to realize that $\compl{\mathcal{H}} \reduces \UWP{\geq}$ via the reduction
\begin{align*}
    r &: \Cmdsk  \times \Store \to  \CmdsCT \times \ExCT \times \ExCT \\
    r &(\qmd,\cst) \triangleq (\cmd_{\qmd}^{\cst},\underline{1},\underline{\sfrac{1}{2}}),
\end{align*}
  with $\cmd_{\qmd,\cst} \triangleq (\pw{x <- \coin(); \IF x \THEN \qmd(}\cst\pw{)\ \ELSE \SKIP})$, for some fresh variable $\x$ of type $\Bool$.

  By definition,
  $\qwp[\cmd_{\qmd,\cst}](\underline{1})(\sigma) = \sfrac{1}{2}$ if $\qmd$ diverges on input store $\cst$ (i.e., if $(\qmd,\cst) \in \compl{\mathcal{H}}$) for any input state $\sigma$.
  On the other hand, if $\qmd$ terminates on $\cst$ then
  $\qwp[\cmd_{\qmd,\cst}](\underline{1})(\sigma) = 1$.
  Conclusively,
  \[
    (\qmd,\cst) \in \compl{\mathcal{H}}
    \iff \forall \sigma \in \StateCT,\ \qwp[\cmd_{\qmd,\cst}](\underline{1})(\sigma) \leq \sfrac{1}{2}
    \iff r(\qmd,\cst) \in \UWP{\leq}
  \]
\end{proof}

\begin{theorem}\label{t:uwp-sub-complete}
  $\UWP{<}$ is $\Pi^{0}_{3}$-complete.
\end{theorem}
\begin{proof}
  From Theorem~\ref{t:wp-sub-complete}, it follows immediately that $\UWP{<} \in \Pi^{0}_{3}$.
For a given expectation $f$, let $f+2$ denote the function defined by $(f+1)(\sigma) \triangleq f(\sigma)+2$, for each state $\sigma \in \StateCT$. 
To show that it is $\Pi^{0}_{3}$-hard, we prove $\compl{\COF} \reduces \UWP{<}$
via the reduction
\begin{align*}
    r &: \Cmdsk \to  \CmdsCT \times \ExCT \times \ExCT \\
    r &(\qmd) \triangleq (\cmd_{\qmd},f_\downarrow,f_\downarrow + 2),
\end{align*}
  where $ \pw{\cmd_\qmd} \triangleq (\pw{ x<- \true; i <- 0; \WHILE x} \DO \{\pw{x <- \coin(); i <- i + 1; \qmd(\cst_{i})}\})$, for some fresh variables $\iv^\Var$ and $\x^\Bool$, and where $f_\downarrow(\cst,\qst) \triangleq \cst(\pw{i})$ is the expectation outputting value of $\iv$ obtained after program execution.
  The program $\cmd_{\qmd}$ executes $\qmd$ with increasing inputs; gradually loosing interest since the loop
  aborts with probability $\sfrac{1}{2}$ after each iteration. Note that the loops body gets stuck if $\pw{\qmd(}\mathit{\cst}_{i}\pw{)}$ is
  non-terminating.

  Recall 
  $\qmd \not\in \COF$, if
  $\qmd$ is a  program in $\Cmdsk$ that is non-terminating on \emph{infinitely many} inputs.
  We claim
  \[
    \qmd \not\in \COF
    \iff \forall \sigma \in \StateCT,\ \qwp[\cmd_{\qmd}](f_\downarrow)(\sigma) < f_\downarrow(\sigma) + 2 \iff r(\qmd) \in \UWP{<}
  \]

  \begin{proofcases}
    \proofcase{
      $\qmd \in \COF$
    }
    In the considered case, there are at most finitely many inputs on which $\qmd$ is non-terminating.
    In particular, this means that $\qmd$ is terminating on all $i \geq \ell$ for some $\ell \in \N$.
    Consider an initial store $\sigma = (\cst,\qst)$ such that $i \triangleq \cst(\iv) \geq \ell$.
    Thus
    \[
      \qwp[\cmd_{\qmd}](f_\downarrow) (\sigma) = \sfrac{1}{2} (i+1) + \sfrac{1}{4} (i + 2) + \sfrac{1}{8}(i+3) + \cdots
      = \sum_{k=1}^{\infty} \frac{i+k}{2^{k}}
      = i + \sum_{k=1}^{\infty} \frac{k}{2^{k}} = i + 2.
    \]
    Conclusively, $(\cmd_{\qmd},f_\downarrow,f_\downarrow+2) \not\in \UWP{<}$.
    \proofcase{
      $\qmd \not\in \COF$
    }
    Let $\sigma = (\cst,\qst)$ be arbitrary, and define $i \triangleq \cst(\iv)$.
    Since in the considered case, there are infinitely many inputs on which $\qmd$ is non-terminating,
    there exists in particular (a minimal) $n \geq 0$ for which $\qmd$ is non-terminating on $\cst_{i + n}$, i.e.,
    the $n$-th iteration of the loops gets stuck.
    Thus
    \[
      \qwp[\cmd_{\qmd}](f_\downarrow)(\sigma) = \sum_{i=1}^{n} \frac{i+k}{2^{k}} < i + 2 = f_\downarrow(\sigma) + 2
    \]
    Conclusively,  $(\cmd_{\qmd},f_\downarrow,f_\downarrow+2) \in \UWP{<}$.
  \end{proofcases}
\end{proof}

\subsection*{Finiteness problems}

\begin{theorem}\label{t:wp-fin-complete}
  $\WP{\neq \infty}$ is $\Sigma^0_2$-complete.
\end{theorem}
\begin{proof}
  For membership, we have
  \[
    (\cmd,\sigma,f) \in \WP{\neq \infty} \iff \exists a \in \RAlg^+,\ (\cmd,\sigma,f,a) \in \WP{<}
  \]
  and thus $\WP{\neq\infty}$ is in $\Sigma^{0}_{2}$, by Theorem~\ref{t:wp-sub-complete}.
  For hardness, we prove $\compl{\mathcal{UH}} \reduces \WP{\neq\infty}$ via the reduction
  \begin{align*}
    r &: \Cmdsk \times \Store \to  \confct \times \ExCT \times \RAlg^+\\
    r & (\qmd,\cst) =  (\cmd_\qmd, \sigma^*,f),
  \end{align*}
  where
  $\cmd_{\qmd}$ is exactly as in Theorem~\ref{t:uwp-sub-complete},
  $\sigma^*$ is an arbitrary but fixed initial state assigning the value $0$ to the fixed variable $\pw{i}$ of type $\Var$, and where $f(s,\rho) \triangleq 2^{s(\iv)}$.
  For correctness of the reduction, first assume $\qmd \not\in\compl{\mathcal{UH}}$, i.e., $\qmd$ is not universally terminating.
  Then, reasoning as in Theorem~\ref{t:uwp-sub-complete},
  \[
    \qwp[\cmd_{\qmd}](f)(\sigma^*) = \sfrac{1}{2} \cdot 2^{1} + \sfrac{1}{4} \cdot 2^{2} + \cdots = \sum_{n = 1}^{\infty} \frac{2^{n}}{2^{n}} = \sum_{n \in \N} 1 = \infty
  \]
  and hence $r(\qmd) \not\in\WP{\neq \infty}$.
  If, on the other hand, $\qmd \in \compl{\mathcal{UH}}$, then
  $\qwp[\cmd_{\qmd}](f)(\sigma^*) = \sum_{n=1}^{k} 1 = k$ where $k$ is the first index with $\qmd$ diverging on $\mathit{\cst}_k$,
  hence $r(\qmd) \in\WP{\neq \infty}$.
\end{proof}

\begin{theorem}\label{t:uwp-fin-complete}
  $\UWP{\neq\infty}$ is $\Pi^0_3$-complete.
\end{theorem}
\begin{proof}
  Membership follows by Theorem~\ref{t:wp-fin-complete}, since we have
  \[
    \cmd \in \UWP{\neq\infty}
    \iff \forall \sigma \in \StateCT,\ (\cmd,\sigma) \in \WP{\neq\infty}
  \]
  For hardness, we prove $\compl{\COF} \reduces \UWP{\neq\infty}$
  via the reduction $r$ defined by $r(\qmd) \triangleq (\cmd_{\qmd},f)$ where $\cmd_{\qmd}$
  is exactly as in Theorem~\ref{t:uwp-sub-complete} but $f(\cst,\qmd) = 2^{\cst(\iv)}$.
  The proof now proceeds identical to Lemma~\ref{t:uwp-sub-complete}.
  If $\qmd \not\in\compl{\COF}$ then
  $\qwp[\cmd_{\qmd}](f) (\cst,\qst) = \sum_{i=1}^{\infty} \frac{2^{n+i}}{2^{i}}=
  \sum_{i=1}^{\infty} 2^{n} = \infty$ considering an initial store
  $(\cst,\qst)$ with $n = \cst(\n)$ so that $\qmd$ is terminating on all $(\cst_{k})_{k \geq n}$.
  Hence $r(\qmd) \not\in \UWP{\neq\infty}$.
  On the other hand, if $\qmd \in\compl{\COF}$ then reasoning identical to Theorem~\ref{t:uwp-sub-complete}
  $\qwp[\cmd_{\qmd}](f)(\sigma)$ is a finitely indexed sum independent of the initial state $\sigma$,
  thus finite. Hence $r(\qmd)\in\UWP{\neq\infty}$.
\end{proof}

\section{Proof of Theorem~\ref{t:adequacy}}

\tadequacy*
\begin{proof}
The adequacy Theorem is a direct consequence of~\cite[Corollary 4.8]{AMPPZ:LICS:22}, based on the observations that $\qwp[\cmd](f)$ corresponds to the notion of expected value $\evalue_\cmd(f) (\sigma)$ of~\cite{AMPPZ:LICS:22} and that $\wpt{\cmd}{f}$ corresponds to the expected value transformer $\qev{\cdot}{\cdot}$ of~\cite{AMPPZ:LICS:22} for the particular Kegelspitze $\CSd = (\Rext,+_{\mathtt{f}})$, with $+_{\mathtt{f}}$ being the forgetful addition.
\end{proof}

\section{Proof of Theorem~\ref{thm:bound}}

\thmbound*

\begin{proof}
  \newcommand{\pp}[2]{\lambda \qst. tr(#1 \qst #1^\dagger)/tr(#2 \qst #2^\dagger)}
  \newcommand{\ppt}[2]{\lambda \qst. \frac{tr(#1 \qst #1^\dagger)}{tr(#2 \qst #2^\dagger)}}
  \newcommand{\qt}[2]{\lambda \qst. \frac{#1 \qst #1^\dagger}{tr(#2 \qst #2^\dagger)}}
  \newcommand{\oM}{M}
  \newcommand{\oN}{N}
  \newcommand{\oO}{O}
  \newcommand{\oP}{P}
  \newcommand{\enc}[1]{\ulcorner #1 \urcorner}
  Suppose $\qinfer{\cmd}{X} = G$, producing side-constraints $C$.
  By definition, every constraint in $C$ is of the form
  \[
    \Gamma \vdash F \leq X_\ell
  \]
  for some $F \in \TERM$, $X_{\ell} \in \VS$ and where $\Gamma$ is wlog. either a Boolean Expression $\bexp \in \BExp$,
  or a probability constraint, i.e., of the form $\proba{i}{k}(\qst) = 0$ or a conjunction thereof.
  %
  Consider the following, semantic preserving, rewrite rules:
  \begin{align*}
    (H_1 +_{p} H_2)[\y := \e] & \to H_1[\y := \e] +_{p} H_2[\y := \e]
    & H[\chi][\y := \e] & \to H[\y := \e][\chi]
    \\
    (H_1 +_{p} H_{2})[\psi] & \to H_1[\psi] +_{p \circ \psi} H_{2}[\psi]
    & H[\chi][\psi] & \to H[\chi \circ \psi]
  \end{align*}
  where $\circ$ denotes function composition. In particular,
  for $\chi = \qt\oM\oN$, $\psi = \qt\oO\oP$ and $p = \ppt\oM\oN$,
  we have
  $\chi \circ \psi = \qt{(\oO\oM)}{(\oO\oN)}$
  and $p \circ \psi = \ppt{(\oM\oO)}{(\oN\oO)}$.
  By exhaustively applying the above rewrite rules, the left-hand side $F$
  normalizes to a barycentric sum over terms of the form
  $X_i[\y_{k_1} := \e_{k_1},\dots, \y_{k_i} := \e_{k_i};\chi_{i}]$, or
  $X_{k}[\overline{\y_{i} := \e_{i}};\chi_{i}]$ for short.
  Let us extend the syntax of terms as follows:
  \begin{align*}
    P,Q & \rgl \lambda \qst. tr(M\qst M^\dagger) \mid 1 - P \mid P\cdot Q \mid P / Q
    &  F,G & \rgl \dots \mid P \cdot F \mid F + G,
  \end{align*}
  where $M \in \matrixspacealg$ is in the Clifford+T fragment,
  with the new connectives are interpreted in the expected way.
  Particularly, $\sem{P} : \densopalg \to \RAlg^{+}$.
  Together with $F_1 +_{p} F_2 \to p \cdot F_1 + (1 - p) \cdot F_2$ the considered constraints
  are now equivalently defined through constraints
  \[
    \Gamma \vdash \sum_{1 \leq i \leq k} P_{i} \cdot X_{i}[\overline{\y_{i} := \e_{i}},\chi_{i}] \leq X_{\ell}
    .
  \]
  Bringing all coefficients $P_{i}$ to a common denominator $Q$, i.e., $P_i = \frac{Q_{i}}{Q}$,
  and by multiplying the left- and right-hand side by $Q$, we finally arrive at a constraint of the form
  \begin{equation}
    \label{eq:constr}
    \Gamma \vdash \sum_{1 \leq i \leq k} Q_{i} \cdot X_{i}[\overline{\y_{i} := \e_{i}},\chi_{i}] \leq Q \cdot X_{\ell}
    \tpkt
  \end{equation}
  Since the denominator $Q$ is guaranteed non-zero under $\Gamma$,
  this constraint is, by construction, semantically equivalent to the initial constraint $\Gamma \vdash F \leq X_{\ell}$.
  Moreover, all coefficients $Q_{i}$ and $Q$ are finite products of traces $\lambda \qst. tr(M_{j}\qst M^\dagger)$,
  for some matrices ${M}_{j}$.

  Recall the encoding of input states in the notion of polynomial expectation~\eqref{eq:polyexp}.
  Let $\vec{Y}$ collect the $n$ variables of $\RAlg^d[\StateCT]$ representing to the classical input,
  similar let $\vec{A}$  (resp. $\vec{B}$) collects the $2^m$ variables representing the real (resp.imaginary part) of the density matrix.
  Let $\vec{Z}$ abbreviate the sequence of all such variables, in total, $n + 2^{2m+1}$ many.
  Let $\mathit{MS}_{d}$ denote the set of monomials over
  $\vec{Z}$ of degree $d$.
  Thus, every polynomial in $\RAlg^d[\StateCT]$ can be written
  as $\sum_{t \in \mathit{MS}_{d}} \mathsf{c}_{t} \cdot t$ for some coefficients $\mathsf{c}_{t} \in \RAlg$.
  In particular, this means that any polynomial solution $\alpha$ to $C$ maps variables
  to polynomials of this form.
  We now represent these polynomial expectations abstractly, introducing for each $X' \in \VS$ and
  $t \in \mathit{MS}_{d}$ an indeterminate coefficient $c_{t,X'}$,
  thereby representing polynomial expectations $\alpha(X')$ as polynomials $P_{X'} \triangleq \sum_{t \in \mathit{MS}_{d}} \mathsf{c}_{t,X'} \cdot t$.

  Using this interpretation of variables, we now turn the constraints~\eqref{eq:constr}
  into polynomial inequalities so as to employ Proposition~\ref{prop:QE}.
  Let $\enc{\e}$ be the translation of program expressions using indeterminates $Y \in \vec{Y}$.
  Clearly, $\enc{\e}$ is a polynomial. Extending this encoding in the natural way,
  we obtain an encoding $\enc{\bexp}$ of Boolean program expressions $\bexp$.
  Observe that any trace $tr(M\qst M^{\dagger})$ is just an algebraic polynomial over the coefficients in $M$ and $M^{\dagger}$.
  As the probability functions $\proba{i}{k}$ just compute such a trace,
  the predicate $\Gamma$ is naturally modeled in the theory of real closed fields,
  using indeterminates $\vec{Y}$ to encode Boolean expressions, and indeterminates
  $\vec{A}$ and $\vec{B}$ the checks on $\proba{i}{k}$.
  Likewise, as we we have eliminated divisions, the factors $Q_{i}$ and $Q$ in \eqref{eq:constr}
  become expressible as algebraic polynomials $\enc{Q_{i}}$
  and $\enc{Q}$ over $\vec{A}$ and $\vec{B}$.

  It remains to encode the terms $X_{i}[\overline{\y_{i} := \e_{i}},\chi_{i}]$, and $X_{\ell}$,
  taking into account their interpretation via $P_{X_{i}}$ and $P_{X_{\ell}}$, respectively.
  First, observe that
  using the encoding for expressions, the update $\overline{\y_{i} := \e_{i}}$ on variable $\y_{j}$ is
  expressible as a polynomial $E_{j}$, given by the composition of corresponding polynomials $\enc{\e_{j}}$.
  Likewise, the real and imaginary parts of entries $(\chi_i(\rho))_{j,k}$
  are expressible as polynomial fractions $\enc{Re(\chi_{j,k})}$
  and $\enc{Im(\chi_{j,k})}$, respectively ($1 \leq j,k \leq 2^m$).
  Putting things together, the constraint~\eqref{eq:constr},
  becomes expressible as
  \begin{align*}
    & \enc{\Gamma}
    \Rightarrow \!\!\sum_{1 \leq i \leq k} \enc{Q_{i}} \cdot
    P_{X_i}[(Y_{i}/E_{i})_{1 \leq i \leq n },(A_{j,k}/\enc{Re(\chi_{j,k})},B_{j,k}/\enc{Im(\chi_{j,k})})_{1\leq j,k \leq 2^{m}}]
    \leq \enc{Q} \cdot P_{X_{\ell}},
  \end{align*}
  The substitution on polynomials $P_{X_i}$ are to be performed in parallel.
  All involved expressions are polynomials, except that the substitutions
  $\enc{Re(\chi_{j,k})}$ and $\enc{Im(\chi_{j,k})}$ introduce fractions of (non-negative) polynomials.
  Multiplying out via a common denominator and rewriting polynomial comparisons to the standard form $P \geq 0$ and $\neg(P \geq 0)$,
  results in a translation of $\enc{\Gamma \vdash F \leq X_{\ell}}$ of
  the initial constraint $\Gamma \vdash F \leq X_{\ell}$.
  Universally quantifying over $\vec{Z}$, the formula becomes equi-satisfiable with $\Gamma \vdash F \leq X_{\ell}$,
  in the sense that any assignment by $\gamma$ from variables $\mathsf{c}_{t,X'}$ to $\RAlg$ makes the corresponding
  assignment $\alpha(X') = P_{X'}[\gamma] = \sum_{t \in \mathit{MS}_{d}} \gamma(\mathsf{c}_{t,X'}) \cdot t$
  a solution to $\Gamma \vdash F \leq X_{\ell}$, and vice versa.

  For decidability,
  it is now sufficient to note that
  $(\cmd,f,\alpha) \in \QINFER(\RAlg^d[\StateCT])$ iff
  (i)~$\alpha$ is a solution to the constraints $C$ generated by $\qinfer{\cmd}{X} = G$;
  (ii)~$\alpha(X) = f$, and
  (iii)~this solution falls into the class $\RAlg^d[\StateCT]$.
  This is an immediate consequence of Theorem~\ref{th:infer-sound}.
  By the construction above, we have already seen that
  (i)~is expressible within the theory of real closed fields.
  Point~(ii) just means that the coefficients of the template $P_{X}$ assigned to the input variable $X$
  coincides with that of $f \in \RAlg^d[\StateCT]$.
  Point~(iii) just means that the polynomials $P_{X'}$ are non-negative.
  Naturally, these constraints have to be fulfilled only on admissible quantum inputs
  (self-adjointness, trace 1, positive eigenvalues),
  which is easily seen to be representable as a formula $\mathsf{admissible}(\vec{A},\vec{B})$.
  Finally, let $\enc{P_{X} = f}$ denote that the algebraic coefficients of $P_{X}$ and $f$ coincide,
  and let $\vec{\mathsf{c}}$ collects all introduced coefficient variables $\mathsf{c}_{{m,X'}}$
  encoding the polynomial assignment $\alpha$.
  Summing up,
  \[
    \textstyle
    \exists \vec{\mathsf{c}},\
    \forall \vec{Y},\ \forall \vec{A},\ \forall \vec{B},\ \mathsf{admissible}(\vec{A},\vec{B})
    \Rightarrow
    \bigwedge_{(\Gamma \vdash F \leq X_{\ell}) \in C}\enc{\Gamma \vdash F \leq X_{\ell}}
    \land \enc{P_{X} = f}
    \land \bigwedge_{X' \in \VS} (P_{X'} \geq 0)
    ,
  \]
  yields a formula in the theory of real closed fields that is valid iff $(\cmd,f,\alpha) \in \QINFER(\RAlg^d[\StateCT])$.
  Through quantifier elimination (Proposition~\ref{prop:QE}) this formula can be turned
  into an equivalent, quantifier free formula $\psi$. Atoms of $\psi$ are just inequalities
  over finite sums and products in $\RAlg$. Since the latter is decidable,
  we conclude decidability of $\QINFER(\RAlg^d[\StateCT])$.

\paragraph{Complexity.} Let $\size{\cmd}$ be the size of the statement $\cmd$. Now we consider the bound $O(|\psi|) \cdot (jd)^{i^{O(l)}}$ in Proposition~\ref{prop:QE} wrt. the above formula:
\begin{itemize}
\item the number of alternations $l$ is constant $l=1$,
\item the degree $d$ of the polynomials is the one we have fixed in $\RAlg^d[\StateCT]$,
\item the number of variables $i$ is equal to the universally quantified variables (for classical and quantum data) $n+2^{2m+1}=2^{O(\size{\cmd})}$ plus the number of existentially quantified variables, which bounded by $\size{\cmd}({n+2^{2m+1}+1})^d$, as the number of monomials in a polynomial of $l_1$ variables and degree $l_2$ is bounded by $(l_1+1)^{l_2}$ and the number of polynomials is bounded by the number of expectation variables (Hence the size of the program $\size{\cmd}$). Hence $i=2^{dO(\size{\cmd})}$,
\item the number of polynomial inequalities $j$ is bounded by $3\size{\cmd}$ (the inequalities added by Figure~\ref{fig:approx}) plus the number of constraints added by the $\mathsf{admissible}(\vec{A},\vec{B})$, which is in $2^{O(\size{\cmd})}$. Hence $j=2^{O(\size{\cmd})}$,
\item the size of the formula $|\psi|$ is also in $2^{O(\size{\cmd})}$ as each polynomial inequality corresponds to a constant number of logical symbols.
\end{itemize}

Putting all together, the obtained bound is in $2^{2^{dO(\size{\cmd})}}$.
\end{proof}

\section{Example of applications of quantum weakest pre-expectation}

\begin{example}
 Consider the expectations $P : \State \to \{0,1\}$ and $f : \State \to \Rext$.
  \begin{enumerate}
    \item If $\cmd$ does not perform any measurement, then $\cmd$ is deterministic,
          in the sense that either $\nf[\cmd](\sigma)=  \{1: \cfgc{\halt}{\tau}\}$,
 	 in case $\cmd$ is terminating on $\sigma$,
          or $\nf[\cmd](\sigma)=  \emptyset$.
          Then $\qwp[\cmd](P) (\sigma) \in \{0,1\}$ gives
          the (weakest) pre-condition for $\tau$ satisfying $P$,
          in the sense of classical (partial) program correctness.
    \item If $\cmd$ performs measurements, then
          $\nf[\cmd](\sigma)$ is a sub-distribution over terminal configurations of the shape $\cfgc{\halt}{\tau}$.
          In this case, $\qwp[\cmd](P)(\sigma) \in [0,1]$ is the probability that these terminal configurations satisfy $P$.
    \item $\qwp[\cmd](\underline{1})(\sigma)$ gives the probability that $\cmd$ terminates on initial state $\sigma$.
          In particular, $\qwp[\cmd](\underline{1}) = \underline{1}$ implies that $\cmd$ is almost-surely terminating.

    \item $\qwp[\cmd](f)(\sigma)$ returns the value $\mathbb{E}_{\nf[\cmd](\sigma)}(f)$, that is, the mean value of the function $f$ on the distribution $\nf[\cmd](\sigma)$.
    \item Given the measurement operators $\{\meas{k}{i}\}_{k \in \{0,1\}}$,  $\qwp[\cmd](f)(\sigma)$, where $f (\cst,\qst) \triangleq tr(\meas{k}{i} \qst)$,  represents the probability that the outcome of measuring the i-th qubit of $\qst$ is  $k$.    \item Suppose the statement $\cmd$ contains a cost counter $\ct$ (some particular variable),
          and let $\ect[\cmd] \triangleq \qwp[\cmd]({cost})$, where the expectation $cost$ is defined by $cost(\cst,\qst) \triangleq \max\{0,\cst(\ct)\}$.
          Then $\ect[\cmd](\sigma) : \Rext$ yields the expected cost of running $\cmd$
          on input $\sigma$, \emph{upon termination}.
          In contrast to \cite{AMPPZ:LICS:22},
          $\ect[\cmd]$ can capture a variety of monotonic \emph{and} non-monotonic cost models such as
          time, space usage, output sizes, bit entanglement, etc.
          Notice that $\ect[\cmd]$ takes only terminating traces into account,
          in the sense of partial correctness. As such, $\ect[\cmd]$
          does not represent the the edge-case of programs
          that are not almost-surely terminating but have finite cost.
  \end{enumerate}
\end{example}

\end{document}